\documentclass[twocolumn,envcountsame]{svjour3}          

\smartqed  

\usepackage{mathptmx}      

\journalname{}
\date{~~}
\def\makeheadbox{}

\usepackage[noproof]{ktmath}
\usepackage{mathtools}

\usepackage{booktabs}
\usepackage{colortbl}
\usepackage{multirow}
\usepackage{amssymb}

\usepackage{xspace}
\usepackage{enumitem}
\setlist{itemsep=0pt,parsep=0pt}             

\usepackage{url}
\makeatletter \def\url@leostyle{\@ifundefined{selectfont}{\def\UrlFont{\sf}}{\def\UrlFont{\scriptsize\ttfamily}}} \makeatother
\usepackage{alltt,color,listings}

\lstdefinelanguage{MyPseudoCode}{
  columns=flexible,
  keepspaces=true,
  basicstyle=\small\ttfamily,
  morekeywords={and,each,else,false,for,fun,if,in,not,return,true,while},
  keywordstyle=\bfseries,
  identifierstyle=\slshape,
  numbers=left,
  numberstyle=\sffamily\scriptsize\color{gray},
  numbersep=5pt,
  morecomment=[l]\#,
  commentstyle=\color{gray},
  mathescape=true,
  escapeinside={(*@}{@*)},
}
\newcommand*{\inv}[0]{\ensuremath{\textit{inv}}}
\newcommand*{\insertOp}[0]{\procName{insert}\xspace}
\newcommand*{\evictOp}[0]{\procName{evict}\xspace}

\newcommand*{\myparagraph}[1]{\medskip\noindent\textbf{#1}}

\newcommand*{\procName}[1]{\mbox{\small\textit{\ttfamily#1}\xspace}}
\newcommand*{\typeName}[1]{\mbox{\small\normalfont\textsf{#1}}\xspace}

\newcommand*{\parAgg}[1]{\ensuremath{\Pi_{#1}^{\otimes}}}
\newcommand*{\identE}[0]{\ensuremath{\textbf{1}}}

\usepackage{graphicx}
\usepackage{graphbox}
\usepackage{tikz}

\usepackage{balance}

\usepackage{microtype}
\begin{document}

\title{In-Order Sliding-Window Aggregation in Worst-Case Constant Time}


\author{Kanat Tangwongsan \and
  Martin Hirzel \and
  Scott Schneider}

\authorrunning{Kanat Tangwongsan, Martin Hirzel, and Scott Schneider}

\institute{K. Tangwongsan,
  Mahidol University, 
  \email{kanat.tan@mahidol.edu}\\
  M. Hirzel,
  IBM Research,
  \email{hirzel@us.ibm.com}\\
  S. Schneider,
  IBM Research,
  \email{scott.a.s@us.ibm.com}
}


\maketitle

\begin{abstract}
  Sliding-window aggregation is a widely-used approach for extracting insights
from the most recent portion of a data stream. The aggregations of interest can
usually be expressed as binary operators that are associative but not
necessarily commutative nor invertible. Non-invertible operators, however, are
difficult to support efficiently.
In a 2017 conference paper, we introduced DABA, the first algorithm for
sliding-window aggregation with worst-case constant time. Before DABA, if a
window had size $n$, the best published algorithms would require $O(\log n)$
aggregation steps per window operation---and while for strictly in-order
streams, this bound could be improved to $O(1)$ aggregation steps on average, it
was not known how to achieve an $O(1)$ bound for the worst-case, which is
critical for latency-sensitive applications.
This article is an extended version of our 2017 paper. Besides describing DABA
in more detail, this article introduces a new variant, DABA Lite, which achieves
the same time bounds in less memory. Whereas DABA requires space for storing
$2n$ partial aggregates, DABA Lite only requires space for $n+2$ partial
aggregates.
Our experiments on synthetic and real data support the theoretical findings.



\end{abstract}

\section{Introduction}\label{sec_intro}

\begin{table}
  \caption{\label{tab_operator_properties}Aggregation operators with their algebraic properties.}
  \centering
  \small
  \begin{tabular}[c]{@{}m{35mm}c@{\hspace*{1.2mm}}c@{\hspace*{1.2mm}}c@{}}
    \toprule
    & Invertible & Associative & Commutative\\
    \midrule
    \textsf{Sum-like}: sum, count, mean, geomean, stddev, ...
    & \checkmark & \checkmark & \checkmark\\
    \midrule[0.3pt]
    \textsf{Max-like}: max, min, argMax, maxCount, M4~\cite{jugel_et_al_2014}, ...
    & $\times$ & \checkmark & $\times$\\
    \midrule[0.3pt]
    \textsf{Mergeable sketch \cite{agarwal_et_al_2012}}:
    Bloom, CountMin, HyperLogLog,
    \mbox{algebraic} classifiers~\cite{izbicki_2013}, ...
    & $\times$ & \checkmark & \checkmark\\
    \hline
  \end{tabular}\par
\end{table}

Stream processing is a now-standard paradigm for handling high-speed continuous
data, spurring the development of many stream-processing
engines~\cite{apache_flink_2016,akidau_et_al_2013,ali_et_al_2011,boykin_et_al_2014,cranor_et_al_2003,hirzel_et_al_2013,kulkarni_et_al_2015,murray_et_al_2013,toshniwal_et_al_2014,zaharia_et_al_2013}.
Since stream processing is often subject to strict quality-of-service or
real-time requirements, it requires low-latency responses. As a mainstay of
stream processing, aggregation (e.g., computing the maximum, geometric mean, or
more elaborate summaries such as Bloom filters~\cite{bloom_1970}) is one of the
most common computations in streaming applications, used both standalone and as
a building block for other analytics. Unfortunately, existing techniques for
sliding-window aggregation cannot consistently guarantee low latency.

Because the newest data is often deemed more pertinent or valuable
than older data, streaming aggregation is typically performed on a
sliding window (e.g., the last hour's worth of data).  This not only
provides intuitive semantics to the end users but also helps bound the
amount of data the system has to keep around.
Following Boykin et al.~\cite{boykin_et_al_2014}, we use the term
aggregation broadly, to include both classical relational aggregation
operators such as sum, geometric mean, and maximum, as well as a more
general class of associative operators.
Table~\ref{tab_operator_properties} lists several such operators and
characterizes them by their algebraic properties.  While some
operators are invertible or commutative, many are not. This paper
focuses on algorithms that work with all associative operators,
including non-invertible and non-commutative ones.

An algorithm for sliding-window aggregation supports three operations
(formally described in Section~\ref{sec_problemdef}):
\procName{insert} for a data item's arrival, \procName{query} for
requesting the current aggregation outcome, and \procName{evict} for a
data item's departure~\cite{hirzel_schneider_tangwongsan_2017}.
This paper presents the \emph{De-Amortized Banker's Aggregator (DABA)},
a novel general-purpose sliding-window aggregation
algorithm that guarantees low-latency response on every operation---in
the worst case, not just on average.  The algorithm is simple and
supports both fixed-sized and variable-sized windows.  It works as
long as (i) the aggregation operator, denoted by $\otimes$ in this
paper, is an associative binary operator and (ii) the window has
first-in first-out (FIFO) semantics.  DABA does not require any other
properties from the $\otimes$ operator. In particular, DABA works
equally well whether $\otimes$ is invertible or non-invertible, commutative
or non-commutative.  DABA supports each of the \procName{query},
\procName{insert}, and \procName{evict} operations by making at most a
constant number of calls to the $\otimes$ operator in the worst case.
This is independent of the window size, denoted by $n$ in this
paper. If each invocation of $\otimes$ takes constant time, then the
DABA algorithm takes worst-case constant time.

We first published the DABA algorithm in a 2015 technical
report~\cite{tangwongsan_hirzel_schneider_2015} and later in a 2017
conference paper~\cite{tangwongsan_hirzel_schneider_2017}.  Prior to
the publication of DABA, the algorithms with the best
algorithmic complexity for this problem in the
published literature took $O(\log n)$
time~\cite{arasu_widom_2004,tangwongsan_et_al_2015}, i.e., not $O(1)$
time like DABA.  After the publication of DABA, there have been
other papers with algorithms that take amortized $O(1)$ time for FIFO
sliding-window aggregation
\cite{shein_chrysanthis_labrinidis_2017,tangwongsan_hirzel_schneider_2019,theodorakis_et_al_2018,villalba_berral_carrera_2019}.
However, these algorithms maintain the $O(1)$ time complexity only in
the amortized sense, i.e., not in the worst case like DABA.  In terms
of space complexity, for a window of size~$n$, DABA stores $2n$
partial aggregates. This journal version of the DABA paper also
introduces a new, previously unpublished algorithm called DABA Lite
that reduces the memory requirements to $n+2$
partial aggregates. Furthermore, this journal version has more
extensive examples and visualizations for our algorithms.

The core idea behind DABA is to start from an algorithm that has amortized
$O(1)$ time complexity and to de-amortize it. The algorithm that serves as a
starting point is Two-Stacks. It uses an old trick from functional programming
for representing a FIFO queue with two stacks, which we augment with
aggregation. A special case of this algorithm was first mentioned in a Stack
Overflow thread~\cite{adamax_2011}. The Two-Stacks algorithm has one rare---but
expensive---operation called \emph{flip} that transfers all data from one stack
to the other. The flip operation causes a latency spike, which can be
undesirable for low-latency streaming. De-amortization turns the average-case
$O(1)$ behavior of Two-Stacks into the worst-case $O(1)$ behavior of DABA by
spreading out the expensive flip operation, thus eliminating the latency spike.

Two-Stacks, DABA, and DABA Lite only work for FIFO windows, i.e., for sliding
windows over in-order streams. Handling out-of-order streams is beyond the scope
of this paper. Indeed, it has been shown that the lower bound on the time
complexity for aggregating out-of-order streams is worse than $O(1)$ unless
disorder in a stream is bounded by a
constant~\cite{tangwongsan_hirzel_schneider_2019}.
Before the emergence of constant-time sliding window aggregation
algorithms, a popular approach for achieving low latency was to use
coarse-grained windows where evictions occur in batches.  This
approach reduces the effective window size $n$, thus making algorithms
whose time complexity depends on $n$ feasible.  Pre-aggregating
inside each batch reduces the cost of aggregating across
batches~\cite{carbone_et_al_2016,krishnamurthy_wu_franklin_2006,li_et_al_2005}.
But coarse-grained windows are an approximation that does not always
satisfy application requirements.
Another popular research topic in sliding-window aggregation is
\emph{window sharing}, where aggregations for multiple different
window sizes are computed on a single data structure. While some
sliding window aggregation algorithms support window sharing in
amortized $O(1)$ time, none of them achieve worst-case $O(1)$
time~\cite{shein_chrysanthis_labrinidis_2017,tangwongsan_hirzel_schneider_2019}.
DABA and DABA Lite achieve worst-case $O(1)$ time but do not support
window sharing.

Experiments show that DABA and DABA Lite perform well in practice.  We have
implemented our new algorithms in C++ and benchmarked them against average-case
$O(1)$ algorithms. True to being worst-case $O(1)$, our results show that DABA
and DABA Lite have lower latency and competitive throughput as we increase the
window size. When the aggregation operation is cheap, the low latency and high
throughput are due to constant-time updates to a lightweight data structure.
When the aggregation operation is expensive, they are due to a low-constant number
of calls to the costly aggregation operator.

Our implementations of DABA, DABA Lite, and all other algorithms used in this
paper are available on GitHub from the open source project Sliding Window
Aggregators\footnote{\url{https://github.com/IBM/sliding-window-aggregators}}.

\section{Problem Definition}\label{sec_problemdef}

This section formalizes the problem of maintaining aggregation in a
first-in first-out (FIFO) sliding window and discusses the kinds of
aggregations supported in this work.

\subsection{Sliding-Window Aggregation Data Type}

This paper is concerned with sliding-window aggregation on in-order
streams with a first-in first-out (FIFO)
window.  In this type of window, the earliest data item to arrive is also the
earliest data item to leave the window.  Hence, the sliding window is
a queue that supports aggregation of its data.
The front of the queue contains the earliest data, the back of the
queue holds the latest data, and the aggregation is
from the earliest to the latest.  As a queue, the window is only affected by two
kinds of changes:
\begin{description}[leftmargin=1em]
\item[\emph{Data Arrival:}] The arrival of a window data item results in a new data item at
  the end of the window. This is often triggered by the arrival of a data item
  in a relevant stream.
\item[\emph{Data Eviction:}] An eviction causes the data item at the front of the
  window to be removed from the window. The choice of when this happens is
  typically controlled by the window policy (e.g., a time-based window evicts
  the earliest data item when it falls out of the time frame of interest and a
  count-based window evicts the earliest data item to keep the size fixed~\cite{gedik_2013}).
  Window eviction policies are orthogonal to the algorithms in this paper.
\end{description}

\noindent
We will model the problem of maintaining aggregation in a FIFO sliding window as an
abstract data type (ADT) with an interface similar to that of a queue.  To
begin, we review an algebraic structure called a monoid:

\myparagraph{Definition:}
  A \emph{monoid} is a triple $\mathcal{M} = (S, \otimes, \identE)$ with
  a binary operator \mbox{$\otimes\!:S \times S \to S$} on $S$ such that
  \begin{itemize}[topsep=2pt, leftmargin=1em]
  \item[--]\emph{Associativity:} For $a, b, c \in S$,
    $a \otimes (b \otimes c) = (a \otimes b) \otimes c$; and
  \item[--]\emph{Identity:} $\identE \in S$ is the identity:
    $\identE \otimes a = a = a \otimes \identE$ for all $a \in S$.
  \end{itemize}
\vspace*{1mm}

In comparison to real-number arithmetic, the $\otimes$ operator can be
seen as a generalization of arithmetic multiplication where the
identity element $\identE$ is a generalization of the number~$1$.
Some of our earlier papers instead used an analogy to arithmetic
addition with an identity element of zero. While that works equally
well, here we adopt the multiplication analogy, because it makes it
natural to adopt a shorthand notation of $abc$ for $a\otimes b\otimes c$.
That shorthand makes it easier to write detailed examples.

A monoid is \emph{commutative} if $a \otimes b = b \otimes a$ for all
$a, b \in S$.  A monoid has a \emph{left inverse} if there exists
a (known and reasonably cheap) function $\inv(\cdot)$ such that
$a \otimes b \otimes \inv(a) = b$ for all $a, b \in S$.  In general, a monoid may
not be commutative nor invertible.

In the context of aggregation, monoids strike a good balance between generality
and efficiency as was demonstrated
before~\cite{boykin_et_al_2014,tangwongsan_et_al_2015,yu_gunda_isard_2009}.  For
this reason, we focus our attention on supporting monoidal aggregation,
formulating the abstract data type as follows:

\myparagraph{Definition:}
  The first-in first-out \emph{sliding-window aggregation} (SWAG) abstract data type
  maintains a collection of window data and supports the following
  operations:
  \begin{itemize}[topsep=2pt,leftmargin=1.25em]
  \item $\procName{query}()$ returns the ordered monoidal product of the
    window data.  That is, if the sliding window contains values   
    $v_0, v_1, \dots, v_{n-1}$ in their arrival order, $\procName{query}$ 
    returns $v_0 \otimes v_1 \otimes \dots \otimes v_{n-1}$. If the window is
    empty, it returns $\identE$.

  \item $\procName{insert}(v)$ adds $v$ to the end of the
    sliding window.  That is, if the sliding window contains values
    $v_0, v_1, \dots, v_{n-1}$ in their arrival order, then
    $\procName{insert}(v)$ updates the collection to
    $v'_0, v'_1, \dots, v'_{n}$, where $v'_i = v_i$ for $i = 0, 1, \dots, n - 1$
    and $v'_{n} = v$.

  \item $\procName{evict}()$ removes the oldest item from the front of the
    sliding window. That is to say, if the sliding window contains values
    $v_0, v_1, \dots, v_{n-1}$ in their arrival order, then
    $\procName{evict}()$ updates the collection to
    $v'_0, v'_1, \dots, v'_{n-2}$, where $v'_i = v_{i+1}$ for
    $i = 0, 1, 2, \dots, n - 2$.

  \end{itemize}

\noindent
Throughout, $n$ will denote the size of the current sliding window and $v_0,
v_1, \dots, v_{n-1}$ will denote the contents of the sliding window in their
arrival order, where $v_0$ is the oldest element. SWAG itself is not a concrete
algorithm; it is merely an abstract data type, defining a set of operations with
their expected behavior. The algorithms introduced in this paper (including
Two-Stacks, DABA, and DABA Lite) are all concrete instantiations for the SWAG
abstract data type.

\subsection{Aggregation on Monoids}

Despite their simplicity, monoids are expressive enough to capture most basic
aggregations~\cite{boykin_et_al_2014,tangwongsan_et_al_2015}, as well
as more sophisticated aggregations such as maintaining approximate membership via
a Bloom filter~\cite{bloom_1970}, maintaining an approximate count of distinct
elements~\cite{flajolet_et_al_2007}, maintaining the versatile
count-min sketch~\cite{cormode_muthukrishnan_2005}, and indeed all 
operators in Table~\ref{tab_operator_properties}.

However, many aggregations (e.g., standard deviation) are not
themselves monoids but can be couched as operations on a monoid with the help of
two extra steps.  To accomplish this, prior work~\cite{tangwongsan_et_al_2015}
gives a framework for the developer to provide three types \typeName{In},
\typeName{Agg}, and \typeName{Out}, and write three functions as follows:
\begin{itemize}[topsep=2pt,leftmargin=1.25em]
\item $\procName{lift}(e: \typeName{In}): \typeName{Agg}$ takes an element of
  the input type and ``lifts'' it to an aggregation type that will
  be monoid operable.
\item $\procName{combine}(v_1: \typeName{Agg}, v_2: \typeName{Agg}): \typeName{Agg}$
  is a binary operator operating on the aggregation type.  In our paper's
  terminology, \procName{combine} is the monoidal binary operator $\otimes$.
\item $\procName{lower}(v: \typeName{Agg}): \typeName{Out}$ turns an element of
  the aggregation type into an element of the output type.
\end{itemize}

\noindent
Consider the example of maintaining the \textsf{maxcount}, which yields the
number of times the maximal value occurs in the window. Define the type
\typeName{Agg} as a pair \mbox{$\langle m,c\rangle$} comprising the maximum $m$
and its count $c$. Then, define the three functions \procName{lift},
\procName{combine}, and \procName{lower} as:

\begin{align*}
  \procName{lift}(e) &=\langle m\mapsto e, c\mapsto 1\rangle\\
  \procName{combine}(v_1,v_2) &=\left\{
    \begin{array}{l@{\quad\textrm{if}\quad}l}
      v_1 & v_1.m > v_2.m\\
      v_2 & v_1.m < v_2.m\\
      \left\langle
        \begin{array}{l@{\,\mapsto\,}l}
          m & v_1.m,\\
          c & v_1.c+v_2.c
        \end{array}\right\rangle
          & v_1.m = v_2.m\\
    \end{array}\right.\\
  \procName{lower}(v)& =v.c
\end{align*}

It is easy to show that the $\procName{combine}$ function is an
associative binary operator with identify element
\mbox{$\identE=\langle-\infty,0\rangle$}. Consequently,
$\mathcal{M}_{\textsf{maxcount}}=(\typeName{Agg},\procName{combine},\langle-\infty,0\rangle)$
is a monoid.

In this framework, a query is \emph{conceptually} answered as follows.  If the
sliding window currently contains the elements $e_0, e_1, \dots, e_{n-1}$, from
the earliest to the latest, then \procName{lift} derives
$v_i = \procName{lift}(e_i)$ for $i = 0, 1, 2, \dots, n-1$.
Then, \procName{combine}, rendered as infix $\otimes$, computes
$v = v_0 \otimes v_1 \otimes \dots \otimes v_{n-1}$. Finally,
\procName{lower} produces the final answer as $\procName{lower}(v)$.

Note that $\procName{lift}$ only needs to be applied to each element when it
first arrives and $\procName{lower}$ to query results at the
end.  Therefore, the present paper focuses exclusively on the issue of
maintaining the monoidal product---i.e., how to make as few invocations
of \procName{combine} as possible.

\subsection{Example Trace}

Using the \textsf{maxcount} monoid mentioned previously as a running example for
the following sections, we will look at a trace of window operations and their
effects on aggregations. Consider a window with the following contents, with the
oldest element on the left and the youngest on the right.

\newcommand*{\wel}[1]{\parbox{4mm}{#1}}

\noindent\hspace*{9mm}\wel{4,}\wel{5,}\wel{3,}\wel{4,}\wel{0,}\wel{4,}\wel{4,}\wel{~~}\wel{~~}\wel{~~}\textsf{max}=5, \textsf{maxcount}=1

The largest number in the window is 5, and it occurs only once, so the
maxcount is~1. The oldest element on the left is 4, which is smaller
than the current maximum~5, so evicting it does not affect the maximum
or the maxcount.

\noindent\hspace*{9mm}\wel{~~}\wel{5,}\wel{3,}\wel{4,}\wel{0,}\wel{4,}\wel{4,}\wel{~~}\wel{~~}\wel{~~}\textsf{max}=5, \textsf{maxcount}=1

If we again evict the oldest element from the left, the maximum
remaining window element becomes~4. Since the number~4 occurs thrice
in the window, the maxcount is~3. The monoid is not invertible: this
update could not have been accomplished by ``subtracting out''
information from the previous partial aggregate.

\noindent\hspace*{9mm}\wel{~~}\wel{~~}\wel{3,}\wel{4,}\wel{0,}\wel{4,}\wel{4,}\wel{~~}\wel{~~}\wel{~~}\textsf{max}=4, \textsf{maxcount}=3

Inserting~2 does not affect the maximum, and hence, it also does not
affect the maxcount.

\noindent\hspace*{9mm}\wel{~~}\wel{~~}\wel{3,}\wel{4,}\wel{0,}\wel{4,}\wel{4,}\wel{2,}\wel{~~}\wel{~~}\textsf{max}=4, \textsf{maxcount}=3

Finally, inserting~6 changes the maximum. Since the newly inserted element is
the only~6 in the window, the maxcount becomes~1.

\noindent\hspace*{9mm}\wel{~~}\wel{~~}\wel{3,}\wel{4,}\wel{0,}\wel{4,}\wel{4,}\wel{2,}\wel{6,}\wel{~~}\textsf{max}=6, \textsf{maxcount}=1

Notice that in this trace, \insertOp and \evictOp do not strictly alternate. In
general, the SWAG data type, as well as all our algorithms, places no
restrictions on how \insertOp and \evictOp may be called. They can be
arbitrarily interleaved, allowing for dynamically-sized windows.


\section{Two-Stacks}\label{sec:twostacks}

\begin{figure*}
  \begin{minipage}{\columnwidth}
    \includegraphics[width=\textwidth]{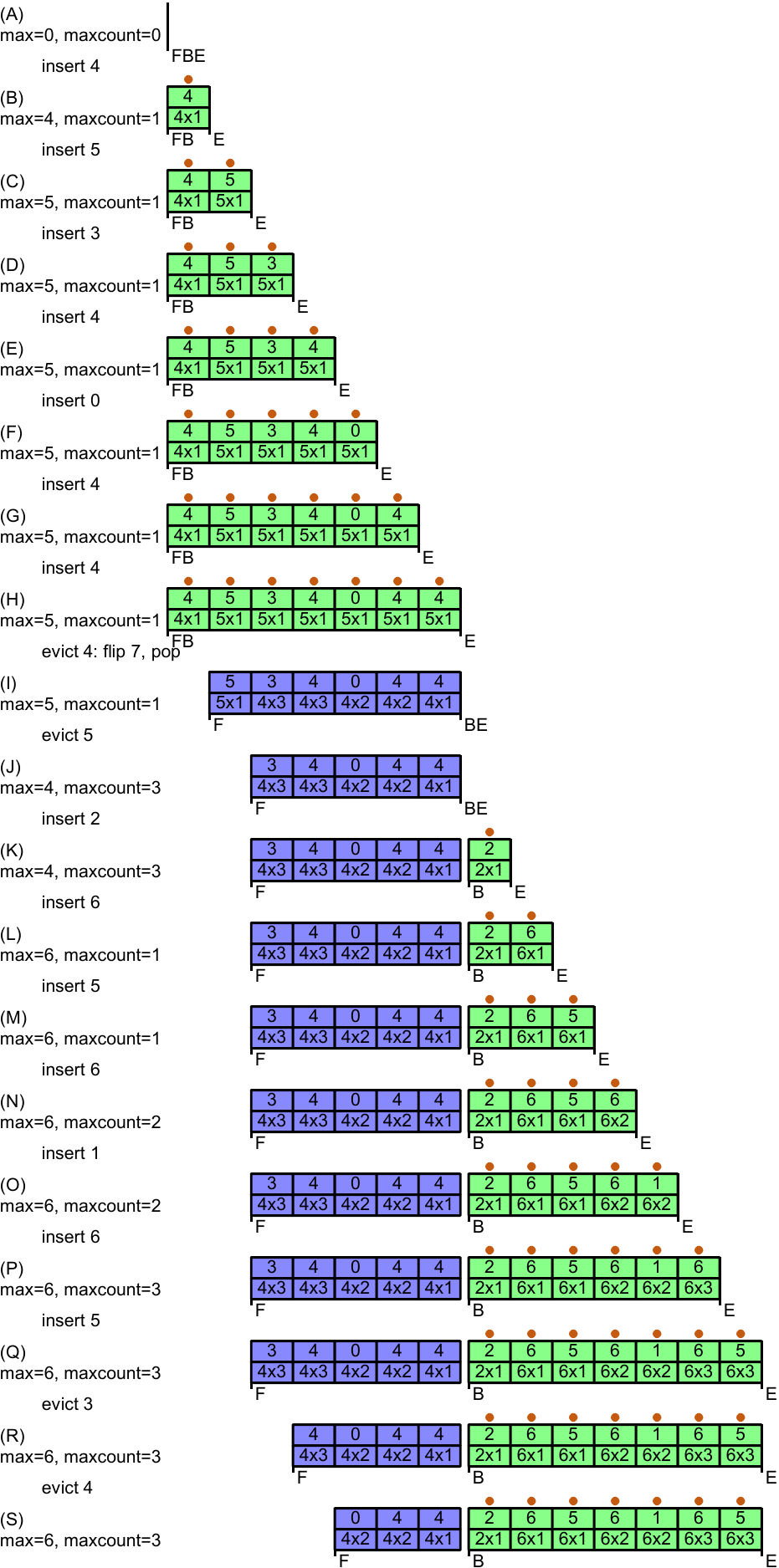}
    \caption{\label{fig_trace_twostacks_maxcount}Two-Stacks example trace for maxcount aggregation. The notation $m\times c$ is shorthand for \textsf{max=}$m$, \textsf{maxcount=}$c$.}
  \end{minipage}\hspace*{\columnsep}\begin{minipage}{\columnwidth}
    \includegraphics[width=\textwidth]{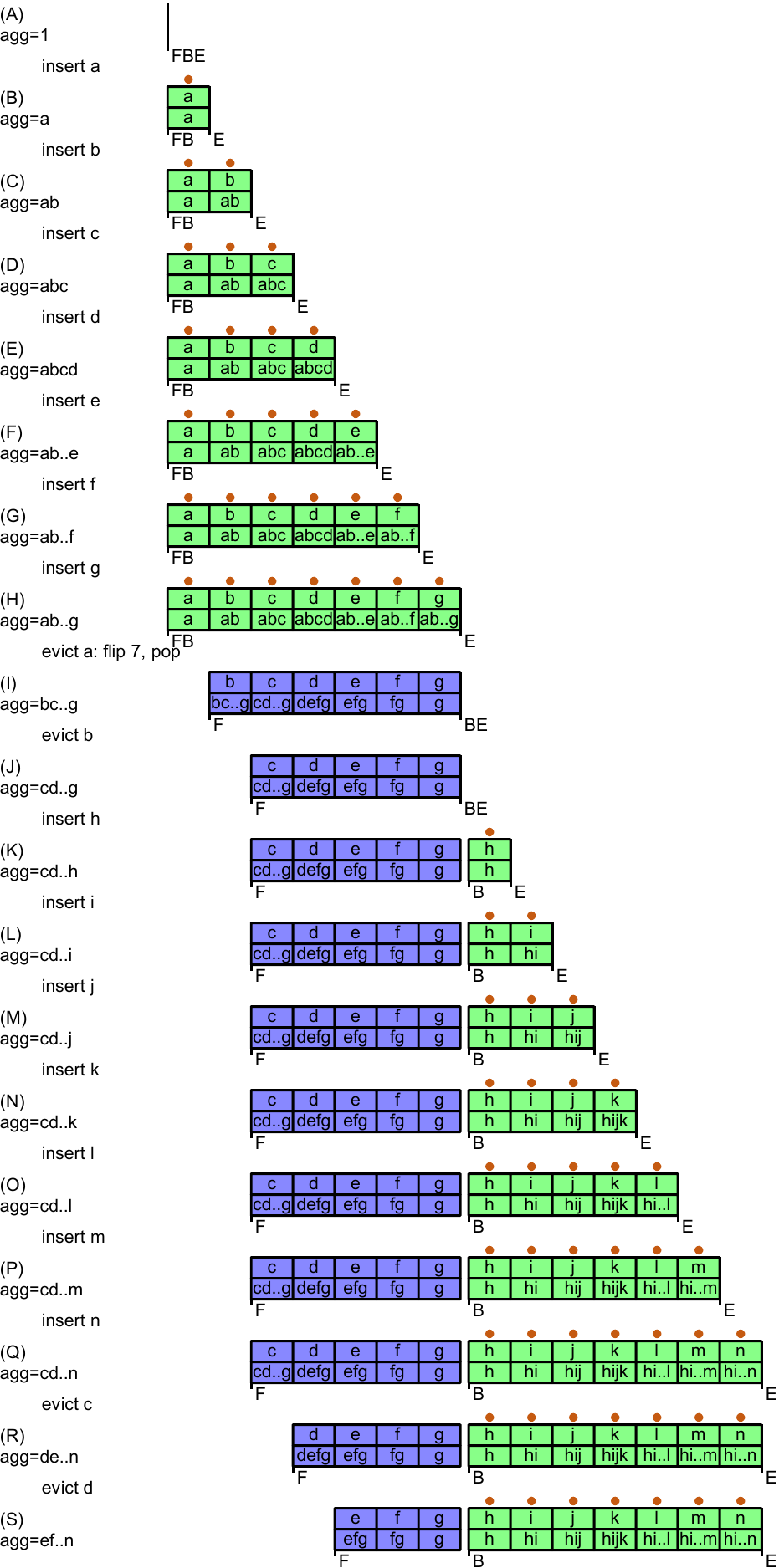}
    \caption{\label{fig_trace_twostacks_concat}Two-Stacks example trace for any aggregation. The notation for aggregates omits $\otimes$, e.g., \textsf{bc..g} is shorthand for $\textsf{b}\otimes\textsf{c}\otimes\textsf{..}\otimes\textsf{g}$.}
  \end{minipage}
\end{figure*}

Two-Stacks is a simple algorithm for in-order sliding window aggregation (SWAG).
For a window size $n$, it stores a total of $2n$ partial aggregates and
implements each SWAG operation with amortized $O(1)$ and worst-case $O(n)$
invocations of~$\otimes$.

The text of this section embeds several data-structure visualizations, which are
all taken from concrete and complete example traces shown in
Figs.~\ref{fig_trace_twostacks_maxcount} and~\ref{fig_trace_twostacks_concat}.

\paragraph{Two-Stacks Data Structure.}

As the name implies, the data structure for Two-Stacks comprises two stacks. We
refer to them as the front stack $F$ and the back stack~$B$. Here is an example
of these two stacks for \textsf{maxcount} aggregation:

\centerline{\includegraphics{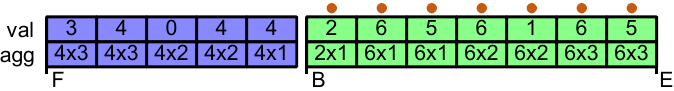}}

\noindent
The front stack is shown in blue rotated $90^\circ$ left, with its top
marked $F$ and its bottom marked~$B$.  The back stack is shown in
green rotated $90^\circ$ right, with its bottom marked $B$ and its top
marked~$E$. To avoid clutter, the visualization only shows $B$ once
to mark the bottoms of both the front stack and the back stack.
As the window slides, evictions pop elements from the
front stack on the left (at~$F$) and insertions push elements on the
back stack on the right (at~$E$). Each stack element is a struct with
two partial aggregates, \procName{val} shown on top
and \procName{agg} shown on the bottom.  The notation \mbox{$4\times
  3$} is shorthand for \textsf{max=}$4$, \textsf{maxcount=}$3$.

For amortized analysis, we use the banker's (aka.~accounting)
method~\cite{clrs_3rded}, which keeps an imaginary savings account. In this
method, every operation is amortized $O(1)$ if we can show that by billing the
user a constant amount for every operation invoked, there is enough money at all
time, without taking out a loan, to pay for the the actual work being done. We
visualize the savings as small golden ``coins'' above the elements; they are not
actual manifest in the data structure.

\paragraph{Two-Stacks Invariants.}

An invariant for a data structure that implements in-order SWAG is a property
that holds before and after every \procName{query}, \procName{insert}, or
\procName{evict}. Let $v_0,\ldots,v_{n-1}$ be the lifted partial aggregates of
the current window contents. Each \procName{val} field stores the corresponding
$v_i$. Each \procName{agg} field stores the partial aggregate of the
corresponding $v_i$ and all other values below it in the same stack. Formally,
if $F[i]$ and $B[i]$ denote the $i^\textrm{th}$ element of $F$ and $B$
indexed from the left starting at index~$0$:\\[-3mm]

\centerline{$\begin{array}{l@{\;}l}
  & \forall i\in 0\ldots |F|-1:\;
    F[i].\procName{val} = v_i\\
  \textrm{and}
  & \forall i\in 0\ldots |F|-1:\;
    F[i].\procName{agg} = v_i\otimes\ldots\otimes v_{|F|-1}\\
  \textrm{and}
  & \forall i\in |F|\ldots |F|+|B|-1:
    B[i-|F|].\procName{val} = v_i\\
  \textrm{and}
  & \forall i\in |F|\ldots |F|+|B|-1:
    B[i-|F|].\procName{agg} = v_{|F|}\otimes\ldots\otimes v_i\\
\end{array}$}

\noindent
The front stack aggregates to the right (easy eviction from the left). The
back stack aggregates to the left (easy insertion from the right).
Here is a visual example of the invariants:

\centerline{\includegraphics{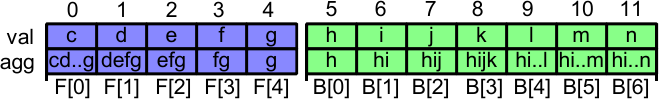}}

\noindent
The notation \typeName{cd..g} is shorthand for
\mbox{$\typeName{c}\otimes\typeName{d}\otimes\ldots\otimes\typeName{g}$}. To
work correctly regardless of commutativity, the aggregation in both stacks is
ordered from the older elements on the left to newer elements on the right. For
example, we take care to aggregate $\typeName{f}\otimes\typeName{g}$ instead of
$\typeName{g}\otimes\typeName{f}$ because \typeName{f} is older than
\typeName{g}.

\paragraph{Two-Stacks Algorithm.}

Being an algorithm that implements in-order SWAG, Two-Stacks needs to
define the functions \procName{query}, \procName{insert}, and
\procName{evict}. But first, we will define two private helper
functions that retrieve the partial aggregate of the entire front
stack $F$ and back stack $B$, respectively.

\begin{lstlisting}[xleftmargin=4mm]
fun $\parAgg{F}$
  if F.isEmpty() return $\identE$ else return F.top().agg
fun $\parAgg{B}$
  if B.isEmpty() return $\identE$ else return B.top().agg
\end{lstlisting}

\noindent
Recall that $\identE$ is the identity element of the monoid. These
helpers return the correct values in constant time, thanks to the
invariants discussed previously. Function \procName{query} combines
the results of both helpers, using a single invocation of~$\otimes$.

\begin{lstlisting}[xleftmargin=4mm,firstnumber=last]
fun query()
  return $\parAgg{F}$ $\otimes$ $\parAgg{B}$
\end{lstlisting}

\noindent
As an example, given the following data structure state, $\parAgg{F}$
is $4\times 3$ and $\parAgg{B}$ is $6\times 2$, so \procName{query} returns
the maximum~$6\times 2$.

\noindent\includegraphics{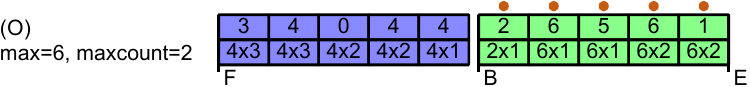}

\noindent
Next, to define \procName{insert} and \procName{evict}, we can assume
that the invariants hold before function calls and must guarantee that
they hold afterwards. The \procName{insert} function just pushes
onto~$B$, taking constant time. Assuming the invariants hold for the
old top of~$B$, \procName{insert} guarantees the invariants for the
new top of~$B$ by setting its partial aggregate \procName{agg} to
$\parAgg{B}\otimes v$.

\begin{lstlisting}[xleftmargin=4mm,firstnumber=last]
fun insert($v$)
  B.push($v$, $\parAgg{B}$ $\otimes$ $v$)
\end{lstlisting}

\noindent
The following example illustrates \procName{insert}.

\noindent\includegraphics{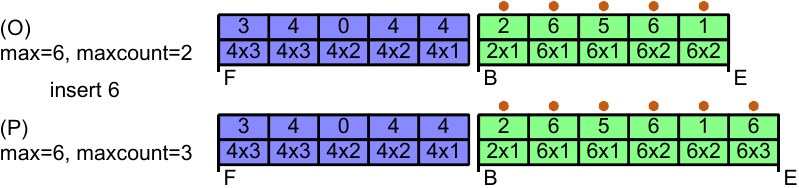}

\noindent
Finally, \procName{evict} pops from $F$ after first ensuring
that $F$ is nonempty.

\begin{lstlisting}[xleftmargin=4mm,firstnumber=last]
fun evict()
  if F.isEmpty()    # Flip
    while not B.isEmpty()
      F.push(B.top().val, B.top().val $\otimes$ $\parAgg{F}$)
      B.pop()
  F.pop()
\end{lstlisting}

\noindent
If $F$ is nonempty, then \procName{evict} is trivial, for example:

\noindent\includegraphics[width=\columnwidth]{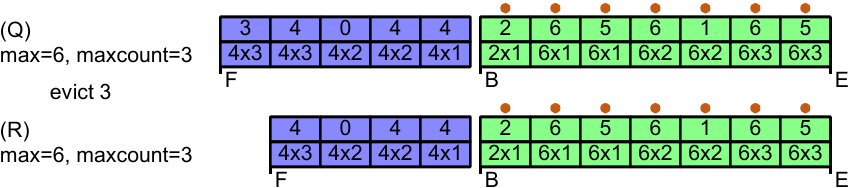}

\noindent
On the other hand, if $F$ is empty, then \procName{evict} must first do a
\emph{flip}. The \emph{flip} operation pushes all values from $B$ onto $F$ and
reverses the direction of the aggregation. In other words, it establishes that
the \procName{agg} fields satisfy the invariant for~$F$. After the flip,
\procName{evict} simply does a \procName{pop} as before.

\noindent\includegraphics[width=\columnwidth]{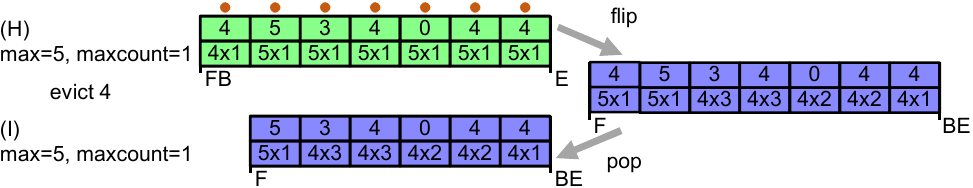}

\noindent
Because of the reversal loop, \procName{flip} takes $O(n)$ time, where $n$ is
the current number of elements. Nevertheless, \procName{evict} only takes
amortized $O(1)$ time, as we will see below.

\paragraph{Two-Stacks Theorems.}

\begin{lemma}\label{trm_twostacks_invariants}
  Two-Stacks maintains the invariants listed above.
\end{lemma}

\begin{proof}
  The \procName{query} function does not modify the stacks and thus
  does not change the invariants. The \procName{insert} function
  maintains the invariants by correctly setting \procName{agg} for the
  newly pushed element. The \procName{evict} function maintains the
  invariants by correctly setting \procName{agg} for all elements of
  $F$ during flip.~\hfill\qed
\end{proof}

\begin{theorem}\label{trm_twostacks_correctness}
  If the window currently contains \mbox{$v_0,\ldots,v_{n-1}$}, then
  \lstinline{query} returns \mbox{$v_0\otimes\ldots\otimes v_{n-1}$}.
\end{theorem}

\begin{proof}
  Using Lemma~\ref{trm_twostacks_invariants},\\[1mm]
\centerline{$\begin{array}{l@{\;}l@{\,\otimes\ldots\otimes\,}l@{\,\otimes\,}l@{\,\otimes\ldots\otimes\,}l}
    \multicolumn{5}{l}{\procName{query}()}\\
  = & \multicolumn{2}{l@{\,\otimes\,}}{\parAgg{F}} & \multicolumn{2}{l}{\parAgg{B}}\\
  = & v_0 & v_{|F|-1} & v_{|F|} & v_{|F|+|B|-1}\\
  = & \multicolumn{4}{l}{v_0\otimes\ldots\otimes v_{n-1}}
\end{array}$}
~\hfill\qed
\end{proof}

\begin{theorem}\label{trm_twostacks_complexity}
  Two-Stacks requires space to store $2n$ partial aggregates.
  Each call to \procName{query} and
  \procName{insert} invokes $\otimes$ exactly one time. Each call to
  \procName{evict} invokes $\otimes$ at most $n$ times and amortized
  $1$ time.
\end{theorem}

\begin{proof}
  The only part of the theorem that is not immediately obvious is the amortized
  complexity of \procName{evict}. To see this, we bill each call to
  \procName{insert} two imaginary coins: one for pushing an element onto~$B$ and
  one to go into the savings. Hence, every element in $B$, as visualized, has a
  golden coin on top of it. When \procName{flip} happens, it invokes $\otimes$
  once for every element of $B$, which is completely paid for by spending the
  coin on that element. Because billing a constant amount per operation covers
  the total cost, each operation is amortized $O(1)$. ~\hfill\qed
\end{proof}

To summarize the workings of Two-Stacks, we will have another look at
Figs.~\ref{fig_trace_twostacks_maxcount} and~\ref{fig_trace_twostacks_concat},
which show complete example traces. Insertions push to the right of the back
stack, visualized in green. Evictions pop from the left of the front stack,
visualized in blue. Given an empty front stack, \procName{evict} first performs
a flip, as show in in \mbox{Step (H)$\to$(I)}. The flip keeps values unchanged
but reverses the associated partial aggregates. In this example, there are 7
partial aggregates to flip, coming from the preceding 7 \procName{insert}
operations, which have deposited 7 coins to the savings to pay for the flip.


\section{Two-Stacks Lite}\label{sec:twostackslite}

\begin{figure*}
  \begin{minipage}{\columnwidth}
    \includegraphics[width=\textwidth]{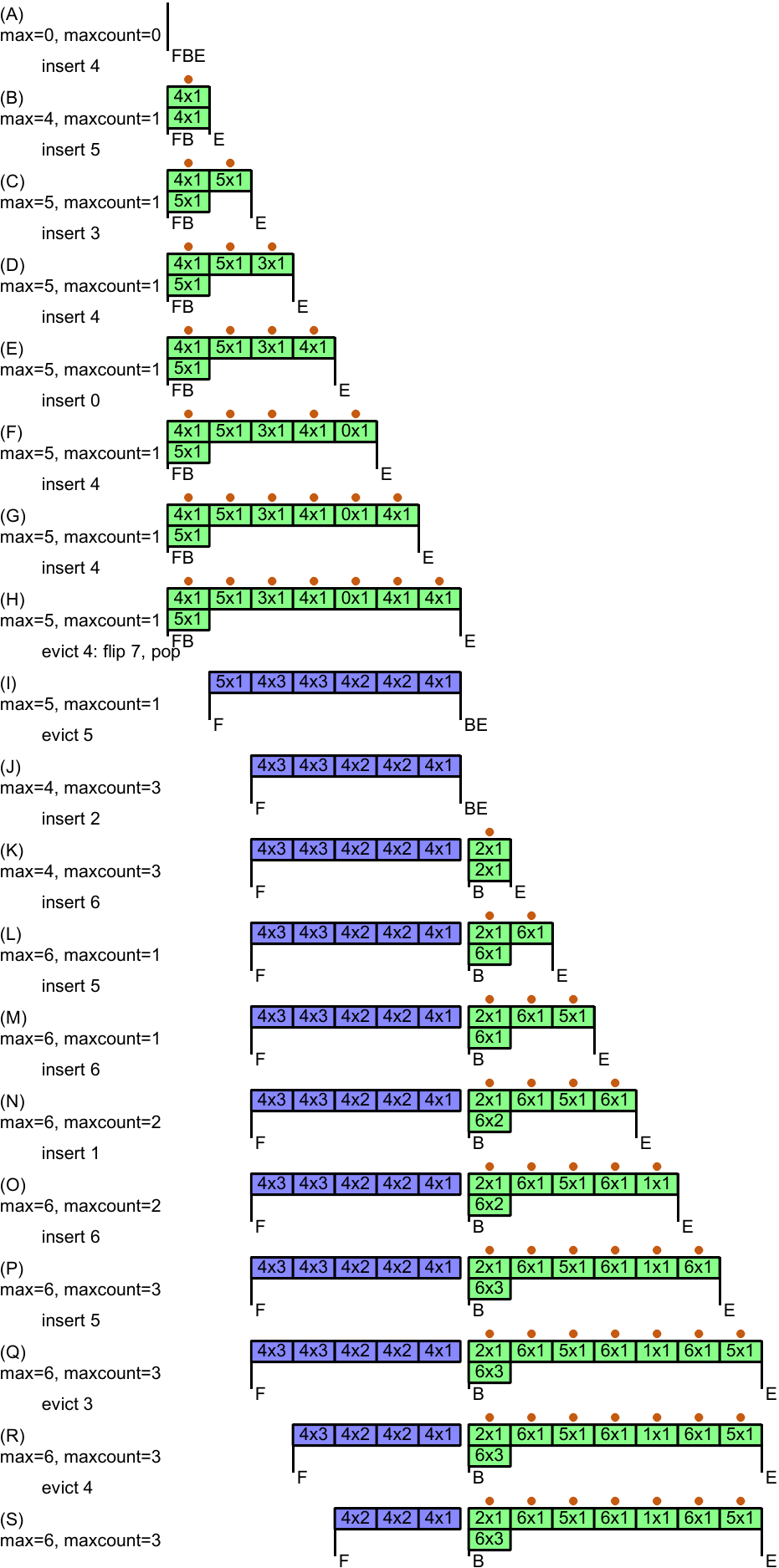}
    \caption{\label{fig_trace_2slite_maxcount}Two-Stacks Lite example trace for maxcount aggregation. The notation $m\times c$ is shorthand for \textsf{max=}$m$, \textsf{maxcount=}$c$.}
  \end{minipage}\hspace*{\columnsep}\begin{minipage}{\columnwidth}
    \includegraphics[width=\textwidth]{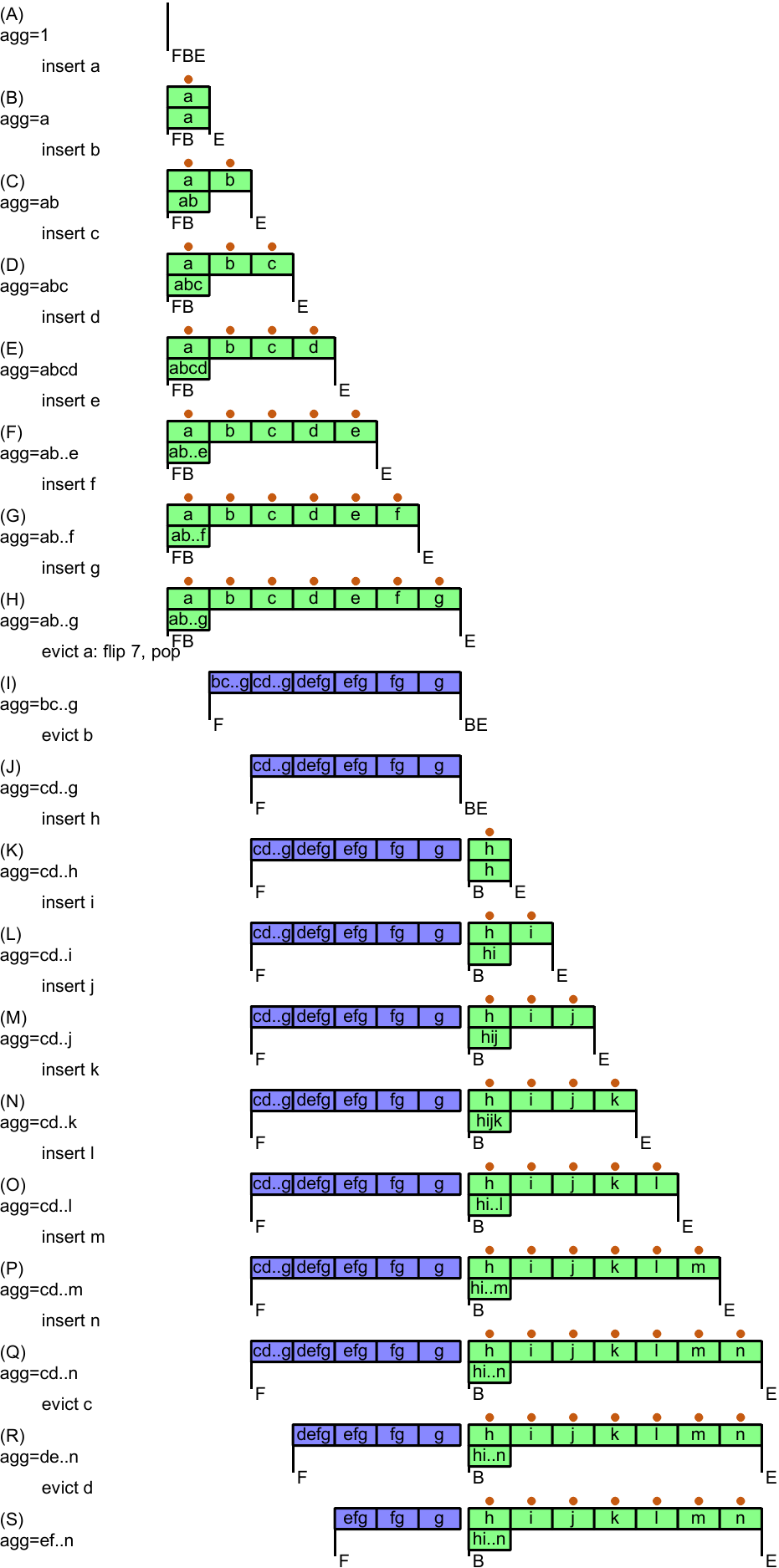}
    \caption{\label{fig_trace_2slite_concat}Two-Stacks Lite example trace for any aggregation. The notation for aggregates omits $\otimes$, e.g., \textsf{bc..f} is shorthand for $\textsf{b}\otimes\textsf{c}\otimes\textsf{..}\otimes\textsf{f}$.}
  \end{minipage}
\end{figure*}

Two-Stacks Lite improves upon Two-Stacks by reducing the space complexity from
$2n$ down to $n+1$ stored partial aggregates. It does this by exploiting the
insight that the Two-Stacks algorithm reads none of the \procName{val} fields of
the front stack and reads only the last \procName{agg} field of the back stack.
This idea comes from the Hammer Slide paper by Theodorakis et
al.~\cite{theodorakis_et_al_2018}. Another improvement is that instead of
physically maintaining two separate stacks, Two-Stacks Lite maintains a single
double-ended queue with an internal pointer $B$ to track the virtual stack
boundary. The time complexity is unchanged, at amortized $O(1)$ and worst-case
$O(n)$ invocations of $\otimes$ per SWAG operation. The data-structure
visualizations in this section are taken from the concrete example
traces shown in Figs.~\ref{fig_trace_2slite_maxcount}
and~\ref{fig_trace_2slite_concat}.

\paragraph{Two-Stacks Lite Data Structure.}

The data structure for Two-Stacks Lite comprises a double-ended queue
\procName{deque} of partial aggregates, one additional partial aggregate
\procName{aggB}, and three pointers $F$, $B$, and $E$. Pointer $F$ points to the
start of \procName{deque}, $B$ points to a location between start and end, and
$E$ points to the end. Here is an example with a max-count aggregation (the
golden ``coins'' visualizing the savings serve the same purpose as in
Two-Stacks):

\centerline{\includegraphics{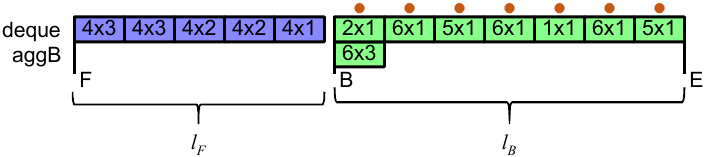}}

\begin{definition}[Pointers]\label{def:pointers}
  In this paper, a \emph{pointer} is an iterator into a resizable
  double-ended queue that supports the following basic data structure
  operations:
  \begin{itemize}
    \item dereference and read or write contents:
      $v\gets\texttt{*}p$, $\texttt{*}p\gets v$
    \item pointer comparison: $p=q$, $p\neq q$
    \item pointer increment and decrement: $p+1$, $p-1$
    \item pointer assignment: $p\gets q$
  \end{itemize}
  Our algorithms only use the above-listed pointer operations. These pointer
  operations can be implemented in worst-case $O(1)$ time over a variable-sized
  double-ended queue by implementing that queue using chunked
  arrays~\cite{tangwongsan_hirzel_schneider_2015}. For stating invariants, we
  will also use a few additional pointer operations such as $p+i$, $p-q$,
  \mbox{or $p<q$}. These additional operations do not need to be $O(1)$ since
  they are used only by invariants and not directly by our
  algorithms.~\hfill\qed
\end{definition}

\noindent
Though physically the data structure uses a single deque, the pointers
demarcate two virtual sublists, $l_F$ and~$l_B$.

\paragraph{Two-Stacks Lite Invariants.}

Let $v_0,\ldots,v_{n-1}$ be the current window contents from the oldest to the
youngest. Then, each deque element in the front sublist $l_F$
\mbox{($F\le p<B$)} stores the partial aggregate starting from the corresponding
$v_i$ up to the element before~$B$. Each deque element in the
back sublist $l_B$ \mbox{($B\le p<E$)} stores the corresponding~$v_i$. In
addition, \procName{aggB} stores the partial aggregate of all elements in $l_B$.
Formally:\\[-3mm]

\centerline{$\begin{array}{l@{\;}l}
  & \forall i\in 0\ldots B-F-1:\;
    \texttt{*}(F+i) = v_i\otimes\ldots\otimes v_{B-F-1}\\
  \textrm{and}~~
  & \forall i\in B-F\ldots E-F-1:\;
    \texttt{*}(F+i) = v_i\\
  \textrm{and}~~
  & \procName{aggB} =
    \texttt{*}B\otimes\ldots\otimes\texttt{*}(E-1)
\end{array}$}

\noindent
As an example, assume that the current window contains the values
$v_0=\typeName{c}$, $v_1=\typeName{d}$, $v_2=\typeName{e}$, etc.\ up
to $v_{10}=\typeName{m}$ and $v_{11}=\typeName{n}$. Then, if there are five
elements in $l_F$ and seven in $l_B$, the data structure looks
as follows:

\centerline{\includegraphics{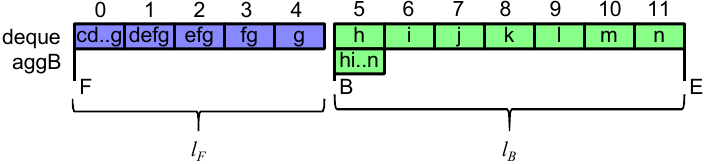}}

\noindent
The notation \typeName{cd..g} is shorthand for
\mbox{$\typeName{c}\otimes\typeName{d}\otimes\ldots\otimes\typeName{g}$}.

\paragraph{Two-Stacks Lite Algorithm.}

The private helper function $\parAgg{F}$ retrieves the partial
aggregate of the entire front sublist $l_F$, where $\identE$ is the
identity element of the monoid.

\begin{lstlisting}[xleftmargin=4mm]
fun $\parAgg{F}$
  if (F = B) return $\identE$ else return *F
\end{lstlisting}

\noindent
Function \procName{query} obtains the partial aggregates of $l_F$
(by calling $\parAgg{F}$) and of $l_B$ (by reading
\procName{aggB}) and combines them in constant time.

\begin{lstlisting}[xleftmargin=4mm,firstnumber=last]
fun query()
  return $\parAgg{F}$ $\otimes$ aggB
\end{lstlisting}

\noindent
As an example, given the following data structure state, $\parAgg{F}$
is $4\times 3$ and \procName{aggB} is $6\times 2$, so \procName{query} returns
the maximum~$6\times 2$.

\noindent\includegraphics{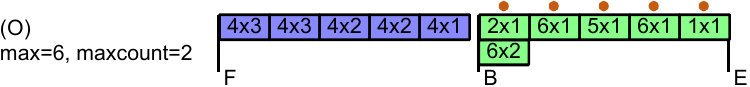}

\noindent
Both \procName{insert} and \procName{evict} can assume that the
invariants hold before they are called and must guarantee that they
hold afterwards. Function \procName{insert}$(v)$ accomplishes this by
pushing $v$ and updating \procName{aggB} accordingly.

\begin{lstlisting}[xleftmargin=4mm,firstnumber=last]
fun insert($v$)
  deque.pushBack($v$)
  E $\gets$ E + 1
  aggB $\gets$ aggB $\otimes$ $v$
\end{lstlisting}

\noindent
The following example illustrates \procName{insert}.

\noindent\includegraphics{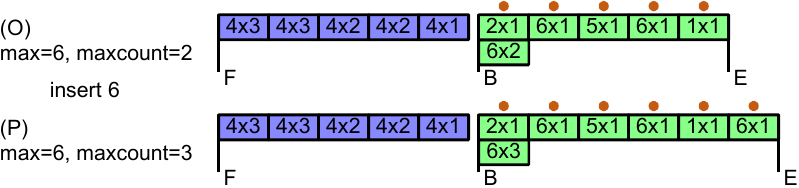}

\noindent
Finally, \procName{evict} pops from the front, after first ensuring
that the front sublist $l_F$ is nonempty.

\begin{lstlisting}[xleftmargin=4mm,firstnumber=last]
fun evict()
  if F = B and B $\neq$ E   # Flip
    I $\gets$ E - 1
    while I $\neq$ F
      I $\gets$ I - 1
      *I $\gets$ *I $\otimes$ *(I + 1)
    B $\gets$ E
    aggB $\gets$ $\identE$
  F $\gets$ F + 1
  deque.popFront()
\end{lstlisting}

\noindent
If $l_F$ is nonempty, then \procName{evict} is trivial, for example:

\noindent\includegraphics[width=\columnwidth]{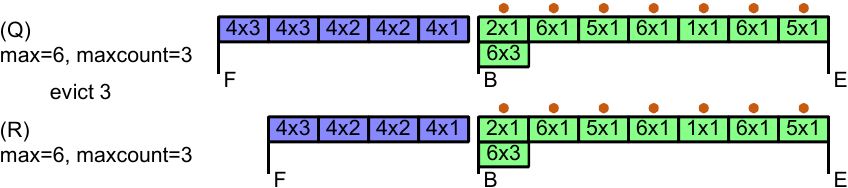}

\noindent
If $l_F$ is empty, then \procName{evict} does a \emph{flip}.
Lines 11--14 rewrite the contents of \procName{deque} in-place to
contain partial aggregates from the corresponding element to the end.
Line~15 updates pointer $B$ to indicate that the front sublist $l_F$
now occupies the entire data structure and the back sublist $l_B$ is
empty. Then, Line~16 resets \procName{aggB} to the monoid identity
element. Here is a visualization of flip followed by pop:

\noindent\includegraphics[width=\columnwidth]{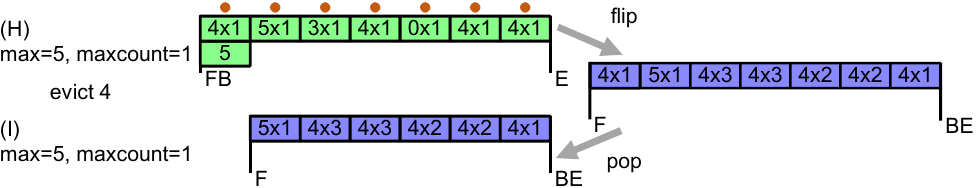}

\noindent
The loop takes time $O(n)$ but can be amortized over the preceding $n$
insertions, where $n$ is the current window size. 

\paragraph{Two-Stacks Lite Theorems.}

\begin{lemma}\label{trm_2slite_invariants}
  The Two-Stacks Lite algorithm maintains the Two-Stacks Lite
  invariants.
\end{lemma}

\begin{proof}
  By inspection, function \procName{query} preserves the data and thus the
  invariants. Function \procName{insert} maintains the invariants by pushing $v$
  and updating \procName{aggB}. Function \procName{evict} optionally does a
  flip, which reestablishes the invariants, then always does a pop, which also
  maintains the invariants.~\hfill\qed
\end{proof}

\begin{theorem}\label{trm_2slite_correctness}
  If the window currently contains \mbox{$v_0,\ldots,v_{n-1}$}, then
  \lstinline{query} returns \mbox{$v_0\otimes\ldots\otimes v_{n-1}$}.
\end{theorem}

\begin{proof}
  Using Lemma~\ref{trm_2slite_invariants},\\[1mm]
\centerline{$\begin{array}{l@{\;}l@{\,\otimes\ldots\otimes\,}l@{\,\otimes\,}l@{\,\otimes\ldots\otimes\,}l}
     \multicolumn{5}{l}{\procName{query}()}\\
  = & \multicolumn{2}{l@{\,\otimes\,}}{\parAgg{F}} & \multicolumn{2}{l}{\procName{aggB}}\\
  = & v_0 & v_{B-F-1} & v_{B-F} & v_{E-F-1}\\
  = & \multicolumn{4}{l}{v_0\otimes\ldots\otimes v_{n-1}}
\end{array}$}
~\hfill\qed
\end{proof}

\begin{theorem}\label{trm_2slite_complexity}
  Two-Stacks Lite requires space to store $n+1$ partial aggregates.
  Each call to \procName{query} and \procName{insert} invokes
  $\otimes$ exactly one time. Each call to \procName{evict} invokes
  $\otimes$ at most $n$ times and amortized one time.
\end{theorem}

\begin{proof}
  Most of the theorem is obvious. To prove the amortized complexity of
  \procName{evict}, bill each call to \procName{insert} two imaginary coins for
  pushing an element onto the back and for the savings (visualized as
  small golden ``coin'' above the elements). Hence, every element in $l_B$ has a
  golden coin on top of it. When \procName{flip} happens, it invokes $\otimes$
  once for every element of $l_B$, which is completely paid for by spending the
  coin on that element.~\hfill\qed
\end{proof}

\section{DABA}\label{sec:daba}

\begin{figure*}
  \begin{minipage}{\columnwidth}
    \includegraphics[width=\textwidth]{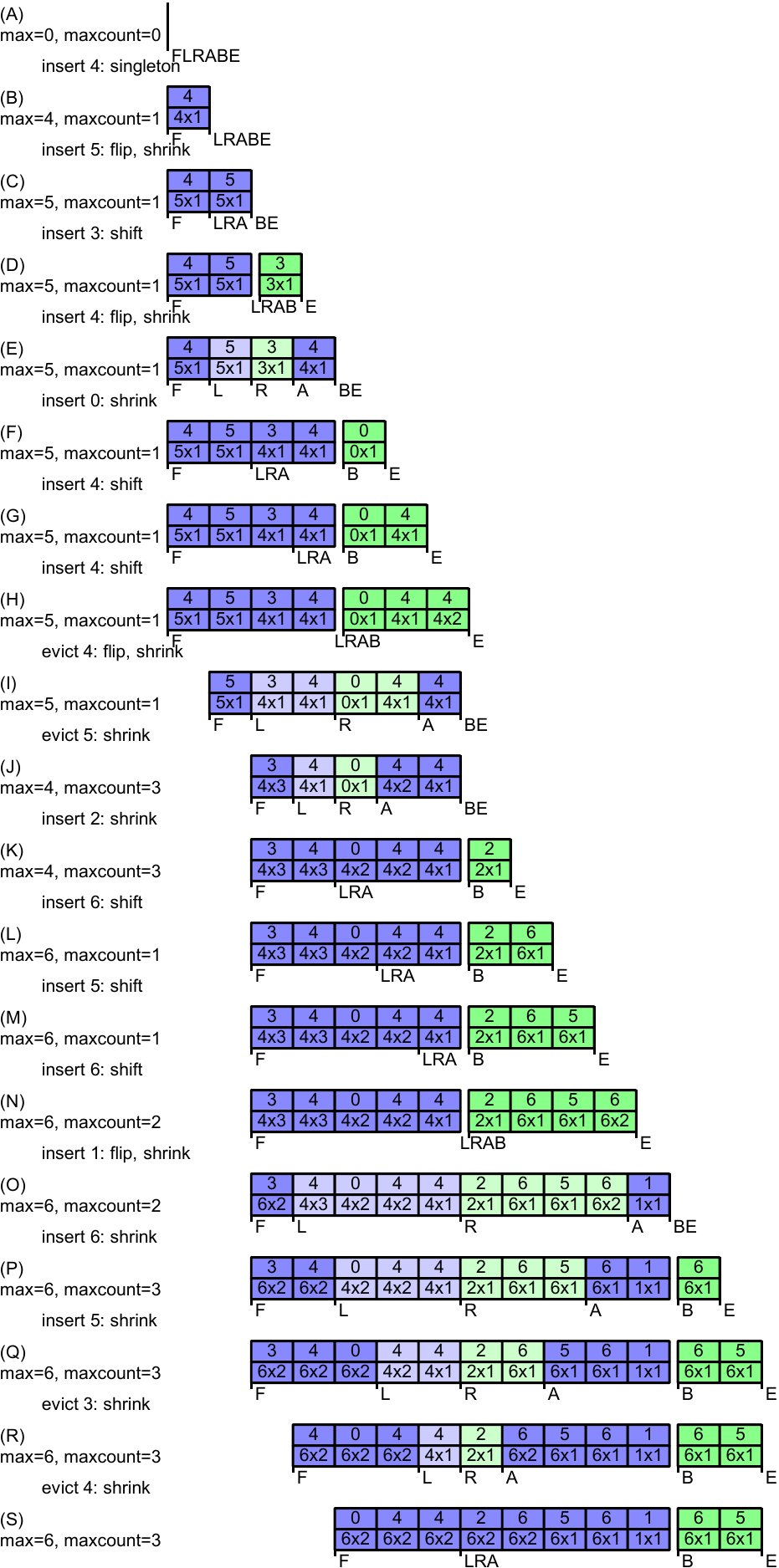}
    \caption{\label{fig_trace_daba_maxcount}DABA example trace for maxcount aggregation. The notation $m\times c$ is shorthand for \textsf{max=}$m$, \textsf{maxcount=}$c$.}
  \end{minipage}\hspace*{\columnsep}\begin{minipage}{\columnwidth}
    \includegraphics[width=\textwidth]{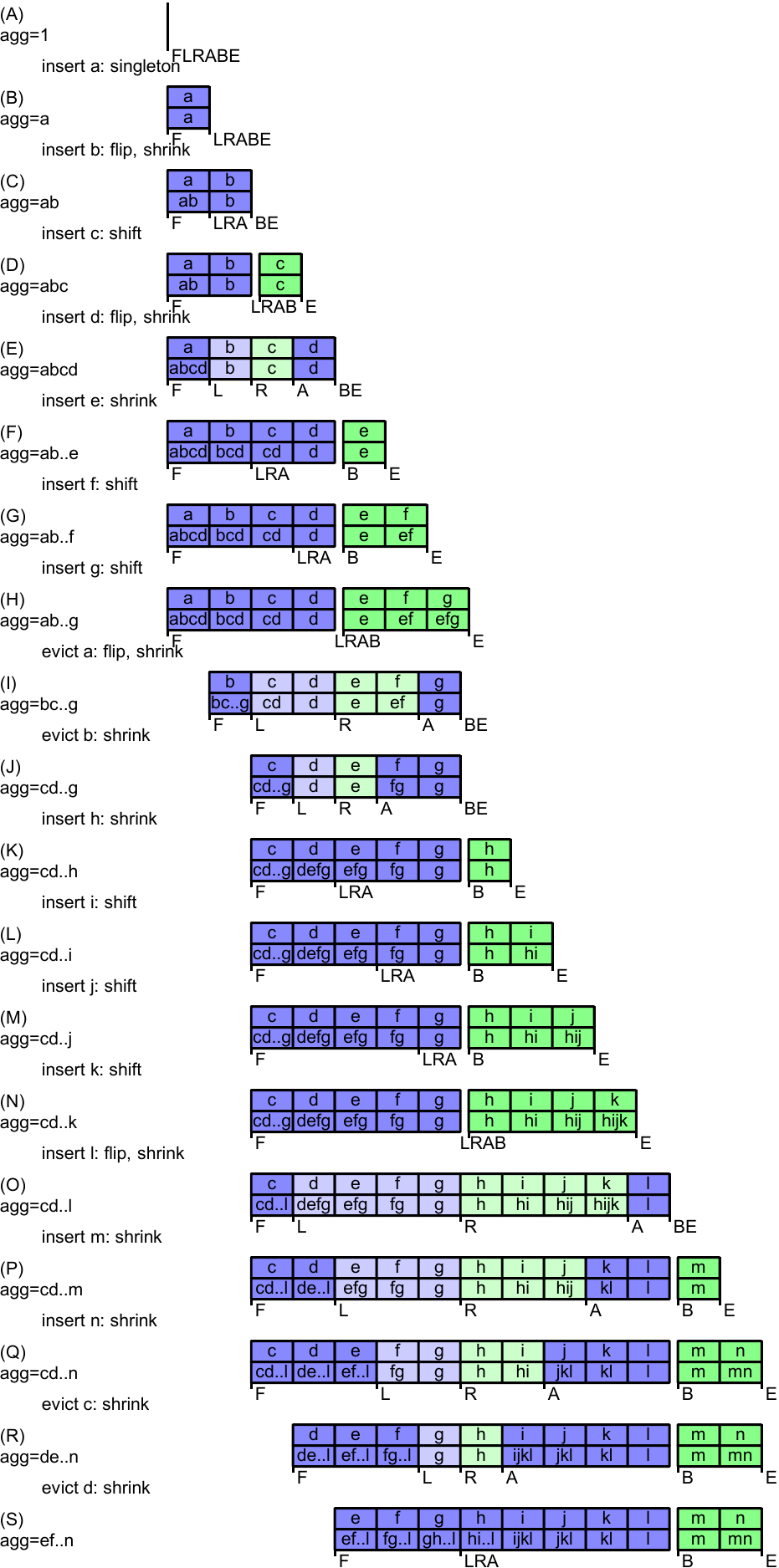}
    \caption{\label{fig_trace_daba_concat}DABA example trace for any aggregation. The notation for aggregates omits $\otimes$, e.g., \textsf{bc..f} is shorthand for $\textsf{b}\otimes\textsf{c}\otimes\textsf{..}\otimes\textsf{f}$.}
  \end{minipage}
\end{figure*}

DABA is a low-latency algorithm for in-order sliding window aggregation (SWAG).
When the window has size $n$, DABA requires space to store $2n$ partial
aggregates and supports each SWAG operation using worst-case $O(1)$ invocations
of~$\otimes$. For a brief explanation of the name, DABA stands for De-Amortized
Banker's Aggregator: Amortization looks at the average cost of an operation over
a long period of time. The banker's method conceptualizes amortization as moving
imaginary coins between the algorithm and a fictitious bank. Deamortization is a
method that turns the average-case behavior into the worst-case behavior,
usually by carefully spreading out expensive operations. In this spirit, notice
that the expensive operation in the Two-Stacks algorithm is the loop for
reversing the direction of aggregation during flip, paid for by imaginary coins
deposited on preceding insertions. Whereas Two-Stacks does the flip late when
the front stack becomes empty, DABA does the flip earlier, when the front and
back stack reach the same length. Furthermore, instead of doing a reversal loop
at the time of the flip, DABA spreads out the steps for reversing the direction
of aggregation. The text of this section embeds several data-structure
visualizations taken from concrete example traces shown in
Figs.~\ref{fig_trace_daba_maxcount} and~\ref{fig_trace_daba_concat}.

\paragraph{DABA Data Structure.}

The DABA data structure comprises a double-ended queue, \procName{deque}, and
six pointers, $F$, $L$, $R$, $A$, $B$, and~$E$, into that queue. Each queue
element is a struct with two partial aggregates: \procName{val} (top row) and
\procName{agg} (bottom row). The basic pointer operations are the same as
in Definition~\ref{def:pointers} and are easy to implement in
$O(1)$ time. The pointers are always ordered as follows:

\centerline{$F \leq L \leq R  \leq A \leq B \leq E$}

\noindent
Here is an example with a max-count aggregation:

\centerline{\includegraphics{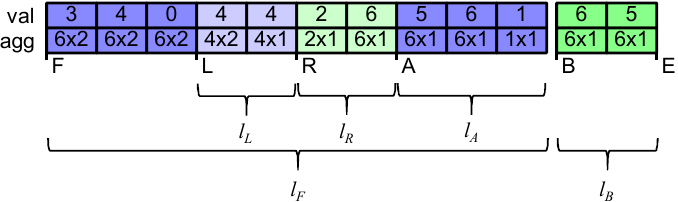}}

\noindent
The \procName{val} fields store the window contents, with the
$i^\mathrm{th}$ oldest value in FIFO order stored at
\mbox{$v_i=*(F+i).\procName{val}$}.
The \procName{agg} fields store partial
aggregations over subranges of the window. Conceptually, each
pointer $p$ corresponds to a sublist~$l_p$. For example, pointer $F$
corresponds to sublist~$l_F$. Each sublist is either aggregated to the
left or to the right. The direction is carefully chosen to enable the
DABA operations.  The leftmost portion of the front list $l_F$ is aggregated to the left to
facilitate eviction. The back list $l_B$ is aggregated to the right to
facilitate insertion. The inner sublists $l_L$, $l_R$, and $l_A$ are
designed to facilitate incremental reversal.  Incremental
reversal happens by adjusting the pointers demarcating sublist
boundaries one step at a time. When a pointer moves, a deque
element changes membership from one sublist to another and its
\procName{agg} field may need to be updated accordingly.

\paragraph{DABA Invariants.}

DABA maintains three groups of invariants: values invariants, partial 
aggregate invariants, and size invariants.
DABA's \textit{values invariants} specify that the \procName{val} field
of each element stores a singleton partial aggregate~$v_i$ obtained by
lifting the corresponding single stream element.

\centerline{$\forall i\in 0\ldots E-F-1 : *(F+i).\procName{val} = v_i$}

\noindent
DABA's \textit{partial aggregate invariants} specify the contents of the
\procName{agg} fields before and after each SWAG operation, based on
sublists. In the visualizations, blue indicates aggregation to the
right and green indicates aggregation to the left.
In the left-most portion of sublist $l_F$ (the front sublist, in dark
blue), each \procName{agg} field holds an aggregate starting from
that element to the right end of~$l_F$.
In sublist $l_L$ (the left sublist, in light blue), each
\procName{agg} field holds an aggregate starting from that element to
the right end of~$l_L$.
In sublist $l_R$ (the right sublist, in light green), each
\procName{agg} field holds an aggregate starting from the left end of
$l_R$ to that element.
In sublist $l_A$ (the accumulator sublist, in dark blue), each
\procName{agg} field holds an aggregate starting from that element to
the right end of~$l_A$, which coincides with the right end of $l_F$.
Finally, in sublist $l_B$ (the back sublist, in dark green), each
\procName{agg} field holds an aggregate starting from the left end of
$l_B$ to that element.
Formally:

\centerline{$\begin{array}{@{}ll@{:\,}l@{}}
           & \forall i\in   0\ldots L\!-\!F\!-\!1 & *(F+i).\procName{agg} = v_i\otimes\ldots\otimes v_{B-F-1}\\
\text{and} & \forall i\in L\!-\!F\ldots R\!-\!F\!-\!1 & *(F+i).\procName{agg} = v_i\otimes\ldots\otimes v_{R-F-1}\\
\text{and} & \forall i\in R\!-\!F\ldots A\!-\!F\!-\!1 & *(F+i).\procName{agg} = v_{R-F}\otimes\ldots\otimes v_i\\
\text{and} & \forall i\in A\!-\!F\ldots B\!-\!F\!-\!1 & *(F+i).\procName{agg} = v_i\otimes\ldots\otimes v_{B-F-1}\\
\text{and} & \forall i\in B\!-\!F\ldots E\!-\!F\!-\!1 & *(F+i).\procName{agg} = v_{B-F}\otimes\ldots\otimes v_i\\
\end{array}$}

Here is a visual example of the invariants:

\centerline{\includegraphics{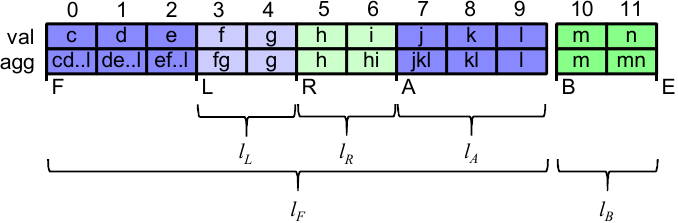}}

\noindent
The notation \typeName{cd..l} is shorthand for
\mbox{$\typeName{c}\otimes\typeName{d}\otimes\ldots\otimes\typeName{l}$}.
To work correctly irrespective of whether the monoid is commutative or
not, in all sublists, the operands of $\otimes$ are always ordered from
older on the left to newer on the right.

DABA's \textit{size invariants} specify constraints on the sizes of
sublists.  Given a pointer $p$, we use the notation $|l_p|$ to
indicate the size of sublist $l_p$.  For example, the size of sublist
$l_F$ is \mbox{$|l_F|=B-F$} and the size of sublist $l_L$ is
\mbox{$|l_L|=R-L$}. Formally, the size invariants are

\vspace{1mm}
\centerline{\small $\Big(|l_F|=0 \; \wedge \; |l_B|=0\Big)
  \vee \Big(\textcolor{blue}{|l_L|+|l_R|+|l_A|+1=|l_F|-|l_B|}
  \; \wedge \; \textcolor{purple}{|l_L|=|l_R|}\Big)$}
\vspace{1mm}

\noindent
This says that the window is either empty ($|l_F|=0$ and $|l_B|=0$) or the
following two conditions hold:

\begin{itemize}[label={-},leftmargin=1em,topsep=0pt]
\item \emph{First}, $\textcolor{blue}{|l_L|+|l_R|+|l_A|+1 \; = \; |l_F|-|l_B|}$.
  The size of the
  front list $l_F$ exceeds the size of the back list $l_B$ by the
  total size of the sublists $l_L$, $l_R$, and $l_A$ plus one.  The
  sublists $l_L$, $l_R$, and $l_A$ are used for incremental reversal,
  and the algorithm, shown below, shrinks their total size by one on
  each insertion or eviction. When the sublists $l_L$, $l_R$, and
  $l_A$ are empty, the algorithm can make just one more insertion or
  eviction before $l_F$ and $l_B$ reach the same size. At that point,
  the algorithm does a \textit{flip}, which relabels $l_F$ and $l_B$
  into $l_L$ and $l_R$, respectively.
\item \emph{Second}, $\textcolor{purple}{|l_L|=|l_R|}$.
  After each \textit{flip}, $l_L$
  and $l_R$ start out with the same size and then shrink at the same
  pace.
\end{itemize}

\noindent
Below, we will see explanations for how the algorithm maintains these
invariants, using color-coding to recognize corresponding
subequations.

\paragraph{DABA Algorithm.}

For each sublist $l_p$, a private helper function $\parAgg{p}$ retrieves the
corresponding partial aggregate or returns the monoid's identity element
$\identE$ if the sublist is empty. Note that for a given sublist, we retrieve
the partial aggregate in the left-most element if that sublist aggregates to
the right, and the partial aggregate in the right-most element if that sublist
aggregates to the left.

\begin{lstlisting}[xleftmargin=4mm]
fun $\parAgg{F}$
  if (F $=$ B) return $\identE$ else return *F.agg
fun $\parAgg{B}$
  if (B $=$ E) return $\identE$ else return *(E-1).agg
fun $\parAgg{L}$
  if (L $=$ R) return $\identE$ else return *L.agg
fun $\parAgg{R}$
  if (R $=$ A) return $\identE$ else return *(A-1).agg
fun $\parAgg{A}$
  if (A $=$ B) return $\identE$ else return *A.agg
\end{lstlisting}

\noindent
These helpers return the correct values in constant time, thanks to the
invariants defined previously. Function \procName{query} combines the
aggregate of $l_F$ and $l_B$, taking only a single invocation
of~$\otimes$.

\begin{lstlisting}[xleftmargin=4mm,firstnumber=last]
fun query()
  return $\parAgg{F}$ $\otimes$ $\parAgg{B}$
\end{lstlisting}

\noindent
Function \procName{insert} pushes a value $v$ with corresponding
partial aggregate to the back of the deque, then calls a function
\procName{fixup}, defined below, for doing one step of incremental
reversal.

\begin{lstlisting}[xleftmargin=4mm,firstnumber=last]
fun insert($v$)
  deque.pushBack($v$, $\parAgg{B}$ $\otimes$ $v$)
  E $\gets$ E + 1
  fixup()
\end{lstlisting}

\noindent
In our running maxcount example, if $\parAgg{B}=0$ and $v=4$, then the
newly pushed deque element has $\procName{val}=4$ and
$\procName{agg}=4$.

\noindent\includegraphics[width=\columnwidth]{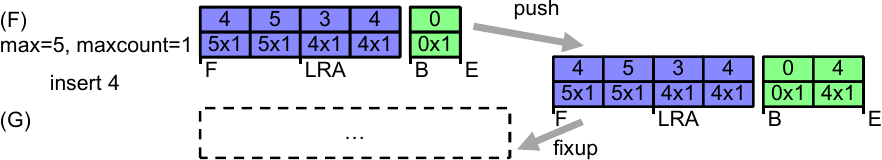}

\noindent
Similarly, \procName{evict} pops an element from the front of the
deque, then calls \procName{fixup} for one step of incremental
reversal.

\begin{lstlisting}[xleftmargin=4mm,firstnumber=last]
fun evict()
  F $\gets$ F + 1
  deque.popFront()
  fixup()
\end{lstlisting}

\noindent
In our running maxcount example, the following picture illustrates
eviction:

\noindent\includegraphics[width=\columnwidth]{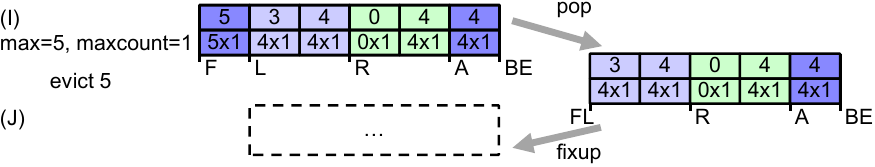}

\noindent
The \procName{fixup} function is responsible for restoring the
invariants. Recall that we can assume that the following size
invariants hold before each call to \procName{insert} or
\procName{evict}:

\vspace{1mm}
\centerline{\small $\Big(|l_F|=0 \; \wedge \; |l_B|=0\Big)
  \vee \Big(\textcolor{blue}{|l_L|+|l_R|+|l_A|+1=|l_F|-|l_B|}
  \; \wedge \; \textcolor{purple}{|l_L|=|l_R|}\Big)$}
\vspace{1mm}

\noindent
Function \procName{insert} grows $l_B$ by one element and function
\procName{evict} shrinks $l_F$ by one element. The impact on the
invariants is the same in both cases: they both decrease the
difference \mbox{$|l_F|-|l_B|$} by one.  Thanks to
the extra $+1$ element in $l_F$, neither \procName{insert} nor
\procName{evict} affects the inner sublists $l_L$, $l_R$,
and~$l_A$. This means that upon entry to \procName{fixup}, the
following is true:

\vspace{1mm}
\centerline{\small $\Big(|l_F|=0 \; \wedge \; |l_B|=1\Big)
  \vee \Big(\textcolor{blue}{|l_L|+|l_R|+|l_A|=|l_F|-|l_B|}
  \; \wedge \; \textcolor{purple}{|l_L|=|l_R|}\Big)$}
\vspace{1mm}

\noindent
Using the above as a precondition, the postcondition of
\procName{fixup} is to reestablish the original size invariants.  The
\procName{fixup} function does this via four cases \emph{singleton},
\emph{flip}, \emph{shift}, and \emph{shrink}.

\begin{lstlisting}[xleftmargin=4mm,firstnumber=last]
fun fixup()
  if F $=$ B        # Singleton case
    B $\gets$ E, A $\gets$ E, R $\gets$ E, L $\gets$ E
  else
    if L $=$ B      # Flip
      L $\gets$ F, A $\gets$ E, B $\gets$ E
    if L $=$ R      # Shift
      A $\gets$ A + 1, R $\gets$ R + 1, L $\gets$ L + 1
    else           # Shrink
      *L.agg $\gets$ $\parAgg{L}$ $\otimes$ $\parAgg{R}$ $\otimes$ $\parAgg{A}$
      L $\gets$ L + 1
      *(A-1).agg $\gets$ *(A-1).val $\otimes$ $\parAgg{A}$
      A $\gets$ A - 1
\end{lstlisting}

The \emph{singleton} case happens when \mbox{$|l_F|=0$}.  Given the
precondition, this can only hold when $|l_B|=1$. Then, without having
to modify the deque, the pointer assignments

\begin{lstlisting}[xleftmargin=4mm,numbers=none]
    B $\gets$ E, A $\gets$ E, R $\gets$ E, L $\gets$ E
\end{lstlisting}

\noindent
change the size of $l_F$ to $|l_F|=1$ while making all the other
sublists $l_L$, $l_R$, $l_A$, and $l_B$ empty, as illustrated below.

\noindent\includegraphics{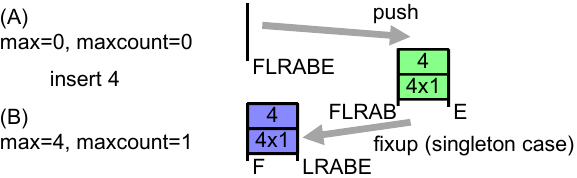}

\noindent
The singleton case code thus establishes

\vspace{1mm}
\centerline{\small
  $\Big(\textcolor{blue}{|l_L|+|l_R|+|l_A|+1=|l_F|-|l_B|=1}
    \; \wedge \; \textcolor{purple}{|l_L|=|l_R|=0}\Big)$}
\vspace{1mm}

\noindent
which implies the original size invariants.

The \emph{flip} case happens when $|l_F|>0$ and the sublists for
incremental reversal are empty:
\mbox{$|l_L|+|l_R|+|l_A|=0$}. Together with the precondition, this
implies that

\vspace{1mm}
\centerline{\small $\Big(\textcolor{blue}{|l_L|+|l_R|+|l_A|=0 \wedge |l_F|=|l_B|}
  \; \wedge \; \textcolor{purple}{|l_L|=0 \wedge |l_R|=0 \wedge |l_A|=0}\Big)$}
\vspace{1mm}

\noindent
Then, the pointer assignments

\begin{lstlisting}[xleftmargin=4mm,numbers=none]
     L $\gets$ F, A $\gets$ E, B $\gets$ E
\end{lstlisting}

\noindent
turn the old outer sublists $l_F$ and $l_B$ into the new inner
sublists $l_L$ and $l_R$, respectively.  No updates to \procName{agg}
fields are required because the corresponding sublists already have
the correct aggregation direction.

\noindent\includegraphics[width=\columnwidth]{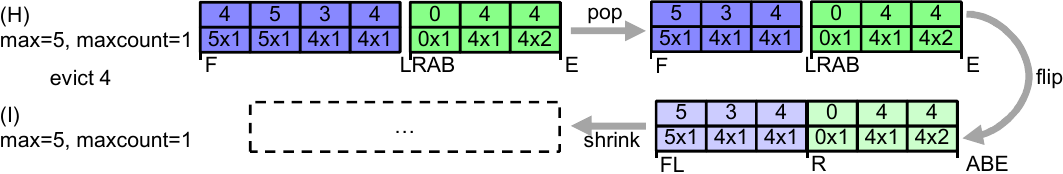}

\noindent
After flip, the following holds:

\vspace{1mm}
\centerline{\small $\Big(\textcolor{blue}{|l_L|+|l_R|+|l_A|=|l_F| \wedge |l_B|=0}
  \; \wedge \; \textcolor{purple}{|l_L|=|l_R|>0}\Big)$}
\vspace{1mm}

\noindent
At this point, we still need to execute the shrink case
to repair the size invariants.

The \emph{shift} case happens when $|l_F|>0$ and $|l_L|=0$. Together
with the precondition, this implies that

\vspace{1mm}
\centerline{\small $\Big(\textcolor{blue}{|l_L|+|l_R|+|l_A|=|l_F|-|l_B|}
  \; \wedge \; \textcolor{purple}{|l_L|=0 \wedge |l_R|=0 \wedge |l_A|>0}\Big)$}
\vspace{1mm}

\noindent
Then, the pointer assignments

\begin{lstlisting}[xleftmargin=4mm,numbers=none]
     A $\gets$ A + 1, R $\gets$ R + 1, L $\gets$ L + 1
\end{lstlisting}

\noindent
increment the pointers separating the left-most portion of $l_F$ from
$l_A$ by one.  No updates to \procName{agg} fields are required
because both the left-most portion of $l_F$ and $l_A$ are governed by the
same \emph{aggs} invariants.

\noindent\includegraphics[width=\columnwidth]{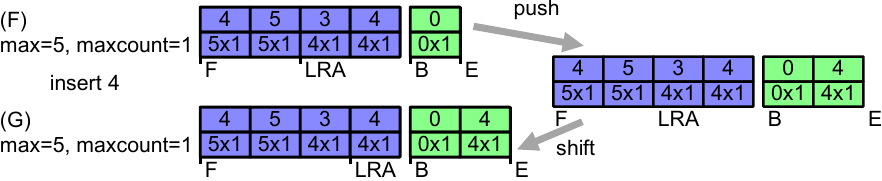}

\noindent
After shift, the following holds:

\vspace{1mm}
\centerline{\small $\Big(\textcolor{blue}{|l_L|+|l_R|+|l_A|+1=|l_F|-|l_B|}
  \; \wedge \; \textcolor{purple}{|l_L|=0 \wedge |l_R|=0}\Big)$}
\vspace{1mm}

\noindent
which implies the original size invariants.

The \emph{shrink} case happens when $|l_F|>0$ and $|l_L|>0$. There are two
scenarios: with or without a flip from the same \procName{fixup}. Either way,
shrink starts with the following precondition:

\vspace{1mm}
\centerline{\small $\Big(\textcolor{blue}{|l_L|+|l_R|+|l_A|=|l_F|-|l_B|}
  \; \wedge \; \textcolor{purple}{|l_L|=|l_R| \wedge |l_L|>0}\Big)$}
\vspace{1mm}

\noindent
The shrink case is the only part of \procName{fixup} that modifies not
just pointers but also \procName{agg} fields. It reduces the sizes of
both $l_L$ and $l_R$ by one each. The top element of $l_L$ becomes
part of the left-most portion of $l_F$, so its \procName{agg} field must
be updated to \mbox{$v_{L-F}\otimes\ldots\otimes
  v_{B-F}=\parAgg{L}\otimes\parAgg{R}\otimes\parAgg{A}$}. The top
element of $l_R$ becomes part of the accumulator sublist $l_A$, so its
\procName{agg} field must be updated to
\mbox{$v_{A-F-1}\otimes\ldots\otimes
  v_{B-F}=v_{A-F-1}\otimes\parAgg{A}$}.

\begin{lstlisting}[xleftmargin=4mm,numbers=none]
      *L.agg $\gets$ $\parAgg{L}$ $\otimes$ $\parAgg{R}$ $\otimes$ $\parAgg{A}$
      L $\gets$ L + 1
      *(A-1).agg $\gets$ *(A-1).val $\otimes$ $\parAgg{A}$
      A $\gets$ A - 1
\end{lstlisting}

\noindent
The following is a typical example of shrink, reducing the sizes of
$l_L$ and $l_R$ from 2 to~1.

\noindent\includegraphics[width=\columnwidth]{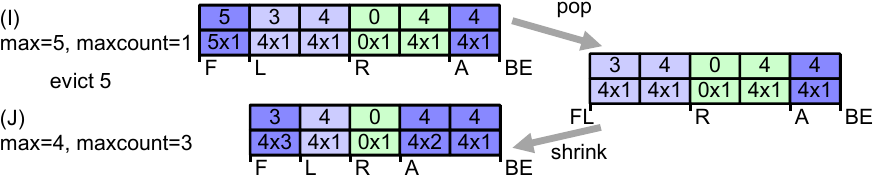}

\noindent
The following is an example of flip followed by shrink. After the
flip, $l_L$ and $l_R$ both have size~3. After the shrink, $l_L$ and
$l_R$ both have size~2.

\noindent\includegraphics[width=\columnwidth]{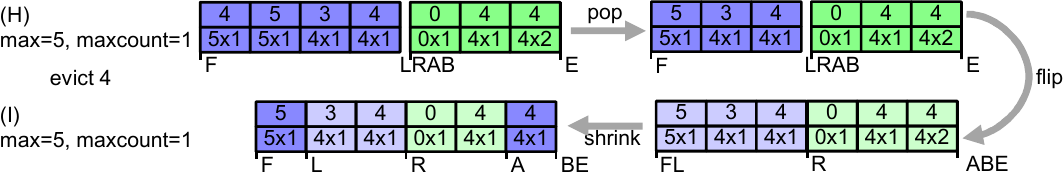}

\noindent
Given the precondition of shrink, it establishes the following
postcondition:

\vspace{1mm}
\centerline{\small $\Big(\textcolor{blue}{|l_L|+|l_R|+|l_A|+1=|l_F|-|l_B|}
  \; \wedge \; \textcolor{purple}{|l_L|=|l_R|}\Big)$}
\vspace{1mm}

\noindent
which implies the original size invariants.

\paragraph{DABA Intuitive View.}
Now that we have seen all the small steps that make up DABA, let us
look at their interplay to help understand the algorithm more
holistically. Figs.~\ref{fig_trace_daba_maxcount}
and~\ref{fig_trace_daba_concat} show two variants of the same example
trace, differing only in their aggregation monoid.  The sequence of
cases starting at (D) comprises $[$flip shrink shrink shift
  shift$]$. A similar pattern starts at (H), comprising $[$flip shrink
  shrink shrink shift shift shift$]$. More generally, each flip is
followed by an equal number of shrink$^m$ and shift$^m$ cases.  For
instance, the sequence starting at (B) is $[$flip shrink shift$]$,
which is a special case \mbox{where $m=1$}. Visually, during the
shrink$^m$ phase of the algorithm, the light blue and light green
sublists narrow to a point, looking like an upside-down step
pyramid. This corresponds to the incremental reversal of $l_R$, which
of course was $l_B$ before the flip. During the shrink$^m$ phase,
pointer $R$ does not change.  Afterwards, during the shift$^m$ phase,
the $LRA$ pointers (which are now all the same) shift to the right one
element at a time until they hit pointer~$B$. When they reach $B$, the
next insert or evict would cause $l_F$ and $l_B$ to have the same
length, triggering the next flip and thus the next cycle. Within each
cycle, pointer $B$ always stays the same; it only moves when a flip
happens.

\paragraph{DABA Theorems.}

\begin{lemma}\label{trm_daba_invariants}
  DABA maintains the invariants listed above, including the values
  invariants, the partial aggregate invariants, and the size invariants.
\end{lemma}

\begin{proof}
  The \procName{query} function does not modify the data structure and
  thus does not change the invariants. Functions \procName{insert} and
  \procName{evict} both establish the same precondition for
  \procName{fixup}, as stated above. Finally, given that precondition, all
  cases of \procName{fixup} reestablish the original invariants as a
  postcondition, as shown above.~\hfill\qed
\end{proof}

\begin{theorem}\label{trm_daba_correctness}
  If the window currently contains \mbox{$v_0,\ldots,v_{n-1}$}, then
  \lstinline{query} returns \mbox{$v_0\otimes\ldots\otimes v_{n-1}$}.
\end{theorem}

\begin{proof}
  Using Lemma~\ref{trm_daba_invariants},\\[1mm]
\centerline{$\begin{array}{l@{\;}l@{\,\otimes\ldots\otimes\,}l@{\,\otimes\,}l@{\,\otimes\ldots\otimes\,}l}
     \multicolumn{5}{l}{\procName{query}()}\\
  = & \multicolumn{2}{l@{\,\otimes\,}}{\parAgg{F}} & \multicolumn{2}{l}{\parAgg{B}}\\
  = & v_0 & v_{B-F-1} & v_{B-F} & v_{E-F-1}\\
  = & \multicolumn{4}{l}{v_0\otimes\ldots\otimes v_{n-1}}
\end{array}$}
~\hfill\qed
\end{proof}

\begin{theorem}\label{trm_daba_complexity}
  DABA requires space to store $2n$ partial aggregates.  DABA invokes
  $\otimes$ at most one time per \lstinline{query}, four times per
  \lstinline{insert}, and three times per \lstinline{evict}.
  Furthermore, for nonempty windows, DABA invokes $\otimes$ on average
  2.5 times per \lstinline{insert} and 1.5 times per
  \lstinline{evict}.
\end{theorem}

\begin{proof}
  The worst-case numbers can be seen directly from the code and by
  noting that the algorithm contains no loops or recursion.  To see
  the average-case numbers, consider the sequence of fixup cases from a
  flip to the next. Immediately following flip, $l_R$ is nonempty and
  $l_A$ is empty. As long as $l_R$ is nonempty, each subsequent
  \procName{insert} or \procName{evict} executes a shrink, invoking
  $\otimes$ three times. When $l_R$ becomes empty, $l_A$ has exactly the
  size that $l_R$ had at the previous flip. As long as $l_A$ is
  nonempty, each subsequent \procName{insert} or \procName{evict}
  executes a shift, without invoking~$\otimes$. The next flip happens
  when $l_A$ is empty. That means that there was an equal number of
  shrink steps as shift steps, and thus, an equal number of
  \procName{fixup} calls with three invocations of $\otimes$ and with
  zero invocations of~$\otimes$. This averages out to 1.5
  $\otimes$-invocations per \procName{fixup}, and thus, 2.5
  $\otimes$-invocations per \procName{insert} and 1.5
  per \procName{evict}.~\hfill\qed
\end{proof}

A corollary of Theorem~\ref{trm_daba_complexity} is that DABA implements all
SWAG operations with worst-case $O(1)$ invocations of~$\otimes$. Unlike previous
algorithms in the paper, DABA involves no costly steps that would require an
amortization argument.


\section{DABA Lite}\label{sec:dabalite}

\begin{figure*}
  \begin{minipage}{\columnwidth}
    \includegraphics[width=\textwidth]{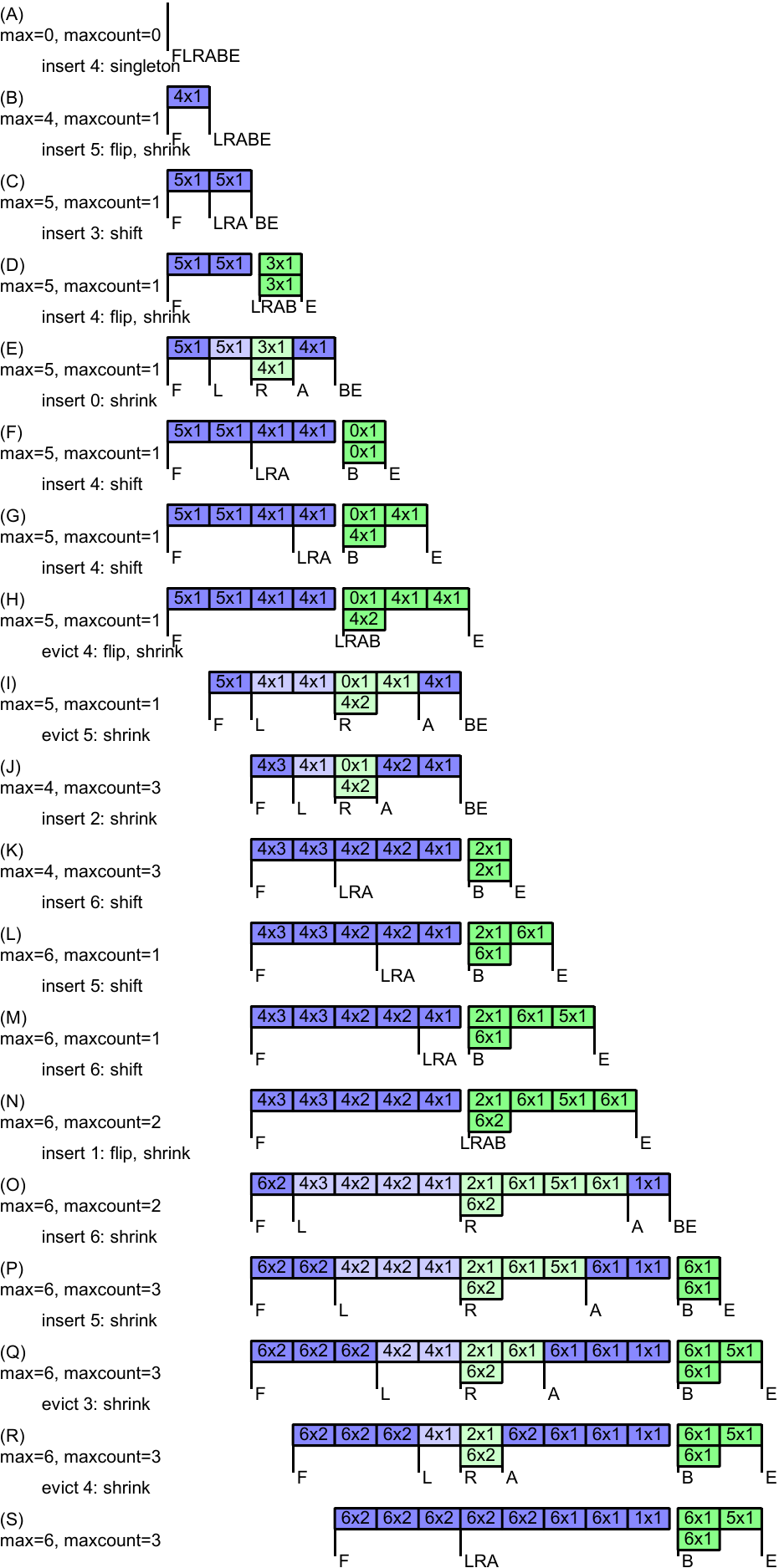}
    \caption{\label{fig_trace_dlite_maxcount}DABA Lite example trace for maxcount aggregation. The notation $m\times c$ is shorthand for \textsf{max=}$m$, \textsf{maxcount=}$c$.}
  \end{minipage}\hspace*{\columnsep}\begin{minipage}{\columnwidth}
    \includegraphics[width=\textwidth]{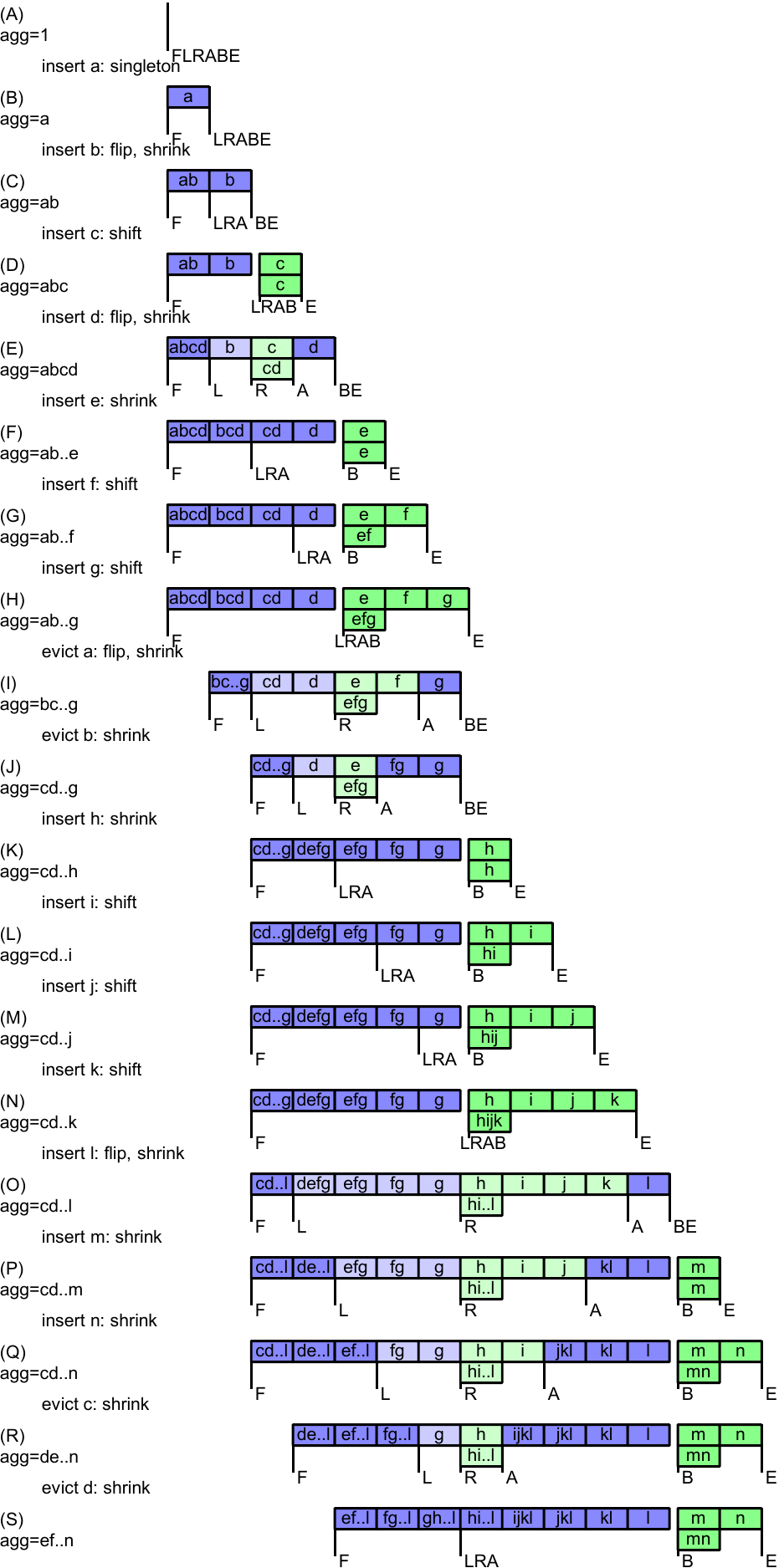}
    \caption{\label{fig_trace_dlite_concat}DABA Lite example trace for any aggregation. The notation for aggregates omits $\otimes$, e.g., \textsf{bc..f} is shorthand for $\textsf{b}\otimes\textsf{c}\otimes\textsf{..}\otimes\textsf{f}$.}
  \end{minipage}
\end{figure*}

DABA Lite improves upon the space complexity of DABA without increasing its
running time, storing only $n+2$ partial aggregates, compared to $2n$ in DABA.
It saves space by exploiting the insights that the DABA algorithm reads none of
the \procName{val} fields of the sublists that are aggregated to the left and
only the last \procName{agg} fields of sublists that are aggregated to the
right. The time complexity is still worst-case $O(1)$ invocations of $\otimes$
per SWAG operation. The data-structure visualizations in this section are all
taken from concrete example traces shown in Figs.~\ref{fig_trace_dlite_maxcount}
and~\ref{fig_trace_dlite_concat}.

\paragraph{DABA Lite Data Structure.}

The DABA Lite data structure comprises a double-ended queue
\procName{deque} of partial aggregates, two additional partial
aggregates \procName{aggRA} and \procName{aggB}, and six pointers $F$,
$L$, $R$, $A$, $B$, and~$E$ into the queue, see
Definition~\ref{def:pointers}. The pointers are always
ordered as follows:

\centerline{$F \leq L \leq R  \leq A \leq B \leq E$}

\noindent
Here is an example with a max-count aggregation:

\centerline{\includegraphics{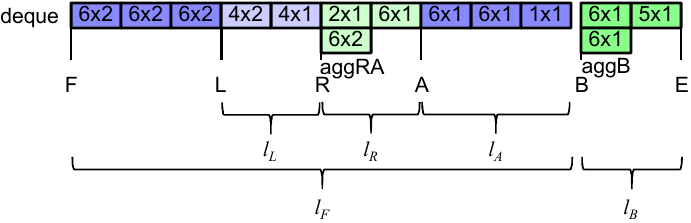}}

\noindent
Conceptually, each pointer $p$ corresponds to a sublist~$l_p$.  Blue
sublists are aggregated to the left to facilitate eviction, with each
element containing the partial aggregate starting from that element to
the right end of its sublist. The elements of green sublists simply
contain the corresponding window elements. The aggregates for the
green sublists are included in \procName{aggRA} and \procName{aggB}.

\paragraph{DABA Lite Invariants.}

The \textit{contents invariants} specify the contents of the
\procName{deque} and of \procName{aggRA} and \procName{aggR}.  Let
$v_0,\ldots,v_{n-1}$ be the current window contents.
In the leftmost portion of sublist $l_F$ (the front sublist, in dark
blue), each element holds an aggregate starting from that element to
the right end of~$l_F$.
In sublist $l_L$ (the left sublist, in light blue), each element holds
an aggregate starting from that element to the right end of~$l_L$.
In sublist $l_R$ (the right sublist, in light green), each element
holds the corresponding window element, and if $L\neq R$ then
\procName{aggRA} holds the
combined partial aggregate of $l_R$ and~$l_A$.
In sublist $l_A$ (the accumulator sublist, in dark blue), each element
holds an aggregate starting from that element to the right end
of~$l_A$.
In sublist $l_B$ (the back sublist, in dark green), each element holds
the corresponding window element, and \procName{aggB} holds the
aggregate of~$l_B$.
Formally:

\centerline{$\begin{array}{@{}ll@{:\,}l@{}}
           & \forall i\in   0\ldots L-F-1 & \texttt{*}(F+i) = v_i\otimes\ldots\otimes v_{B-F-1}\\
\text{and} & \forall i\in L-F\ldots R-F-1 & \texttt{*}(F+i) = v_i\otimes\ldots\otimes v_{R-F-1}\\
\text{and} & \forall i\in R-F\ldots A-F-1 & \texttt{*}(F+i) = v_i\\
\text{and} & \multicolumn{2}{l}{(L=R)\;\vee\;(\procName{aggRA} = v_{R-F}\otimes\ldots\otimes v_{B-F-1})}\\
\text{and} & \forall i\in A-F\ldots B-F-1 & \texttt{*}(F+i) = v_i\otimes\ldots\otimes v_{B-F-1}\\
\text{and} & \forall i\in B-F\ldots E-F-1 & \texttt{*}(F+i) = v_i\\
\text{and} & \multicolumn{2}{l}{\procName{aggB} = v_{B-F}\otimes\ldots\otimes v_{E-F-1}}\\
\end{array}$}

\noindent
Here is a visual example of the invariants for a window with contents
\mbox{$v_0=\procName{c},v_1=\procName{d},\ldots,v_{10}=\procName{m},v_{11}=\procName{n}$}:

\centerline{\includegraphics{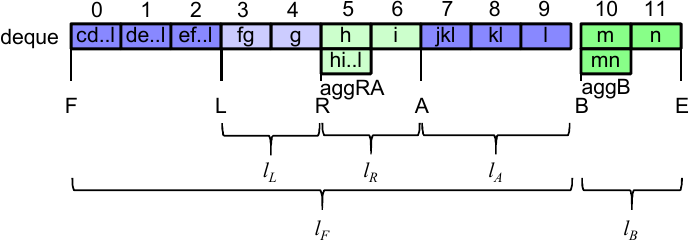}}

\noindent
The notation \typeName{cd..l} is shorthand for
\mbox{$\typeName{c}\otimes\typeName{d}\otimes\ldots\otimes\typeName{l}$}.

The \textit{size invariants} specify constraints on the sizes of
sublists. The size invariants of DABA Lite are the same as those of
DABA:

\vspace{1mm}
\centerline{\small $\Big(|l_F|=0 \; \wedge \; |l_B|=0\Big)
  \vee \Big(\textcolor{blue}{|l_L|+|l_R|+|l_A|+1=|l_F|-|l_B|}
  \; \wedge \; \textcolor{purple}{|l_L|=|l_R|}\Big)$}
\vspace{1mm}

\paragraph{DABA Lite Algorithm.}

For each sublist $l_p$ that is aggregated to the left, a private
helper function $\parAgg{p}$ retrieves the corresponding partial
aggregate or returns the monoid's identity element $\identE$ if the
sublist is empty.

\begin{lstlisting}[xleftmargin=4mm]
fun $\parAgg{F}$
  if (F $=$ B) return $\identE$ else return *F
fun $\parAgg{L}$
  if (L $=$ R) return $\identE$ else return *L
fun $\parAgg{A}$
  if (A $=$ B) return $\identE$ else return *A
\end{lstlisting}

\noindent
These helpers return the correct values in constant time thanks to the
invariants defined previously. Function \procName{query} combines the
aggregate of $l_F$ and $l_B$, taking only a single invocation
of~$\otimes$.

\begin{lstlisting}[xleftmargin=4mm,firstnumber=last]
fun query()
  return $\parAgg{F}$ $\otimes$ aggB
\end{lstlisting}

\noindent
Function \procName{insert} pushes a value $v$ onto $l_B$ and updates
\procName{aggB} accordingly, then calls a function \procName{fixup},
defined below, for doing one step of incremental reversal.

\begin{lstlisting}[xleftmargin=4mm,firstnumber=last]
fun insert($v$)
  deque.pushBack($v$)
  E $\gets$ E + 1
  aggB $\gets$ aggB $\otimes$ $v$
  fixup()
\end{lstlisting}
In our running example, if $procName{aggB}=0\otimes1$ and $v=4$, the
newly pushed deque element is $4\otimes1$ and the updated \procName{aggB} is
also~$4\otimes1$.

\noindent {\includegraphics[width=\columnwidth]{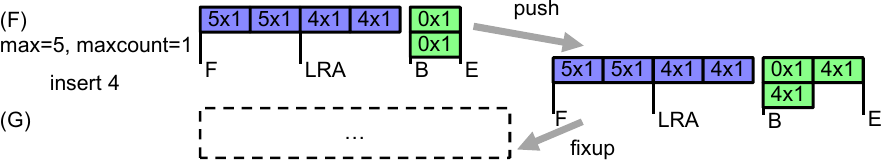}}

Similarly, \procName{evict} pops an element from the front of the
deque, then calls \procName{fixup} for one step of incremental
reversal.

\begin{lstlisting}[xleftmargin=4mm,firstnumber=last]
fun evict()
  F $\gets$ F + 1
  deque.popFront()
  fixup()
\end{lstlisting}
\noindent
For our running example, eviction is illustrated below:

\noindent {\includegraphics[width=\columnwidth]{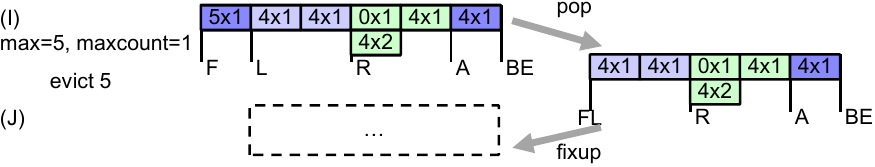}}

\noindent
The \procName{fixup} function repairs the invariants. The effect of
\procName{fixup} on the size invariants is the same for DABA and for
DABA~Lite. Since Section~\ref{sec:daba} has a formal analysis, here we
only have an informal discussion. As before, the \procName{fixup}
function has four cases \emph{singleton}, \emph{flip}, \emph{shift},
and \emph{shrink}.

\begin{lstlisting}[xleftmargin=4mm,firstnumber=last]
fun fixup()
  if F $=$ B    # Singleton case
    B $\gets$ E, A $\gets$ E, R $\gets$ E, L $\gets$ E
    aggRA $\gets$ $\identE$
    aggB $\gets$ $\identE$
  else
    if L $=$ B  # Flip
      L $\gets$ F, A $\gets$ E, B $\gets$ E
      aggRA $\gets$ aggB
      aggB $\gets$ $\identE$
    if L $=$ R  # Shift
      A $\gets$ A + 1, R $\gets$ R + 1, L $\gets$ L + 1
    else       # Shrink
      *L $\gets$ $\parAgg{L}$ $\otimes$ aggRA
      L $\gets$ L + 1
      *(A-1) $\gets$ *(A-1) $\otimes$ $\parAgg{A}$
      A $\gets$ A - 1
\end{lstlisting}

The \emph{singleton} case happens when $|l_F|=0$ and $|l_B|=1$.

\begin{lstlisting}[xleftmargin=4mm,numbers=none]
    B $\gets$ E, A $\gets$ E, R $\gets$ E, L $\gets$ E
\end{lstlisting}

\noindent
The pointer assignments change this to $|l_F|=1$ and $|l_B|=0$.

\noindent\includegraphics{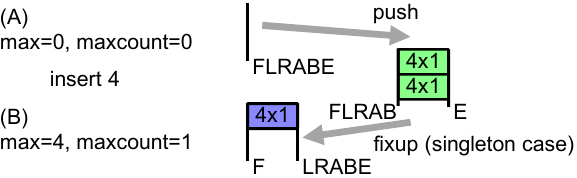}

The \emph{flip} case happens when $|l_F|>0$ and $|l_L|=0$ and
$|l_R|=0$ and $|l_A|=0$. That implies that $|l_F|=|l_B|$, which means
we can simply turn the old outer sublists $l_F$ and $l_B$ into the new
inner sublists $l_L$ and $l_R$.

\begin{lstlisting}[xleftmargin=4mm,numbers=none]
     L $\gets$ F, A $\gets$ E, B $\gets$ E
     aggRA $\gets$ aggB
     aggB $\gets$ $\identE$
\end{lstlisting}

\noindent
The corresponding updates to \procName{aggRA} and \procName{aggB} do
not require any invocations of~$\otimes$.

\noindent\includegraphics[width=\columnwidth]{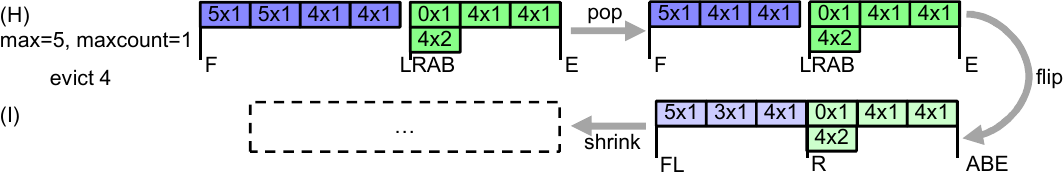}

The \emph{shift} case happens when $|l_F|>0$ and $|l_L|=0$.  That
means the pointers $L=R=A$ are equal, and can simply be moved one
element to the right.

\begin{lstlisting}[xleftmargin=4mm,numbers=none]
     A $\gets$ A + 1, R $\gets$ R + 1, L $\gets$ L + 1
\end{lstlisting}

\noindent
There is no need to update \procName{aggRA}, since it will not be read
anymore until after the next flip. The invariant for \procName{aggRA}
remains satisfied, thanks to~\mbox{$L=R$}.

\noindent\includegraphics[width=\columnwidth]{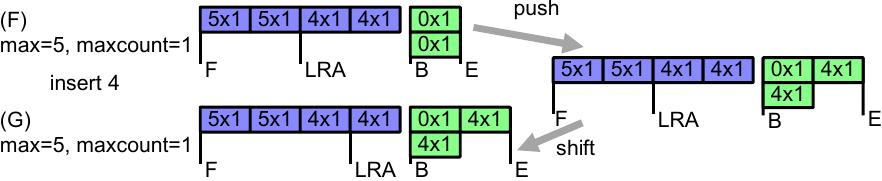}

The \emph{shrink} case happens when $|l_F|>0$ and $|l_L|>0$.  The
sizes of $|l_L|=|l_R|$ are the same, and shrink reduces them by one
each. This requires setting \procName{agg} fields of blue sublists
appropriately for their contents invariants.

\begin{lstlisting}[xleftmargin=4mm,numbers=none]
     *L $\gets$ $\parAgg{L}$ $\otimes$ aggRA
     L $\gets$ L + 1
     *(A-1) $\gets$ *(A-1) $\otimes$ $\parAgg{A}$
     A $\gets$ A - 1
\end{lstlisting}

\noindent
Even though the internal boundary $A$ between $l_R$ and $l_A$ moves,
taken together, these two sublists still occupy the same elements, and
thus, \procName{aggRA} does not change. Consequently, the shrink case
of DABA Lite requires one less $\otimes$-invocation than the shrink
case of DABA.

\noindent\includegraphics[width=\columnwidth]{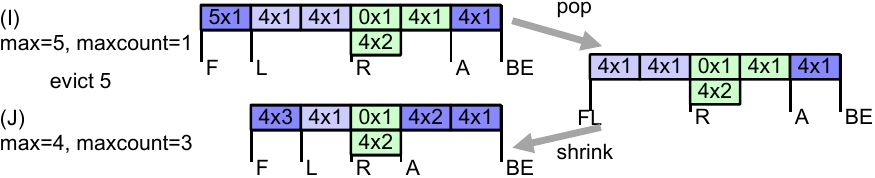}

\noindent
The shrink case also gets triggered right after a flip.

\noindent\includegraphics[width=\columnwidth]{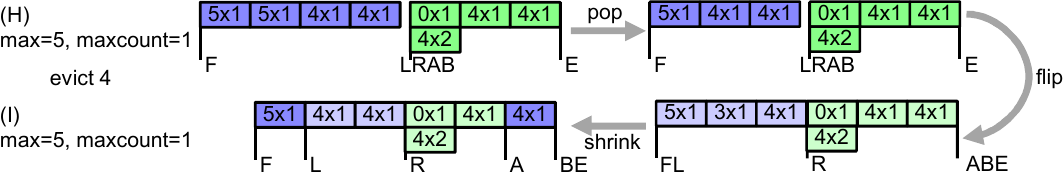}

\paragraph{DABA Lite Theorems.}

\begin{lemma}\label{trm_dabalite_invariants}
  DABA Lite maintains both the contents invariants and the size
  invariants defined above.
\end{lemma}

\begin{proof}
  The \procName{query} function does not modify the data structure and thus does
  not change the invariants. The effect of \procName{insert}, \procName{evict},
  and \procName{fixup} on the size invariants is the same as for DABA. Whenever
  sublist boundaries change, the code updates the contents of \procName{deque},
  \procName{aggRA}, and \procName{aggB}, if necessary, to uphold the contents
  invariants.~\hfill\qed
\end{proof}

\begin{theorem}\label{trm_dabalite_correctness}
  If the window currently contains \mbox{$v_0,\ldots,v_{n-1}$}, then
  \lstinline{query} returns \mbox{$v_0\otimes\ldots\otimes v_{n-1}$}.
\end{theorem}

\begin{proof}
  Using Lemma~\ref{trm_dabalite_invariants},\\[1mm]
\centerline{$\begin{array}{l@{\;}l@{\,\otimes\ldots\otimes\,}l@{\,\otimes\,}l@{\,\otimes\ldots\otimes\,}l}
\multicolumn{5}{l}{\procName{query}()}\\
  = & \multicolumn{2}{l@{\,\otimes\,}}{\parAgg{F}} & \multicolumn{2}{l}{\procName{aggB}}\\
  = & v_0 & v_{B-F-1} & v_{B-F} & v_{E-F-1}\\
  = & \multicolumn{4}{l}{v_0\otimes\ldots\otimes v_{n-1}}
\end{array}$}
~\hfill\qed
\end{proof}

\begin{theorem}\label{trm_dabalite_complexity}
  DABA Lite requires space to store $n+2$ partial aggregates.  DABA
  Lite invokes $\otimes$ at most one time per \lstinline{query}, three
  times per \lstinline{insert}, and two times per \lstinline{evict}.
  Furthermore, for nonempty windows, DABA Lite invokes $\otimes$ on
  average two times per \lstinline{insert} and one time per
  \lstinline{evict}.
\end{theorem}

\begin{proof}
  The algorithm contains no loops or recursion, so we can directly see
  the worst-case numbers from the code.  The average-case numbers are
  based on the observation that every sequence of shrink steps is
  followed by an equal number of shift steps. Shrink requires two
  $\otimes$-invocations and shift requires zero $\otimes$-invocations,
  so the average \procName{fixup} call has one
  $\otimes$-invocation.~\hfill\qed
\end{proof}

\section{Experimental Evaluation}\label{sec_results}

%
%

\begin{figure*}[!t]
\center
\includegraphics[width=0.33\textwidth]{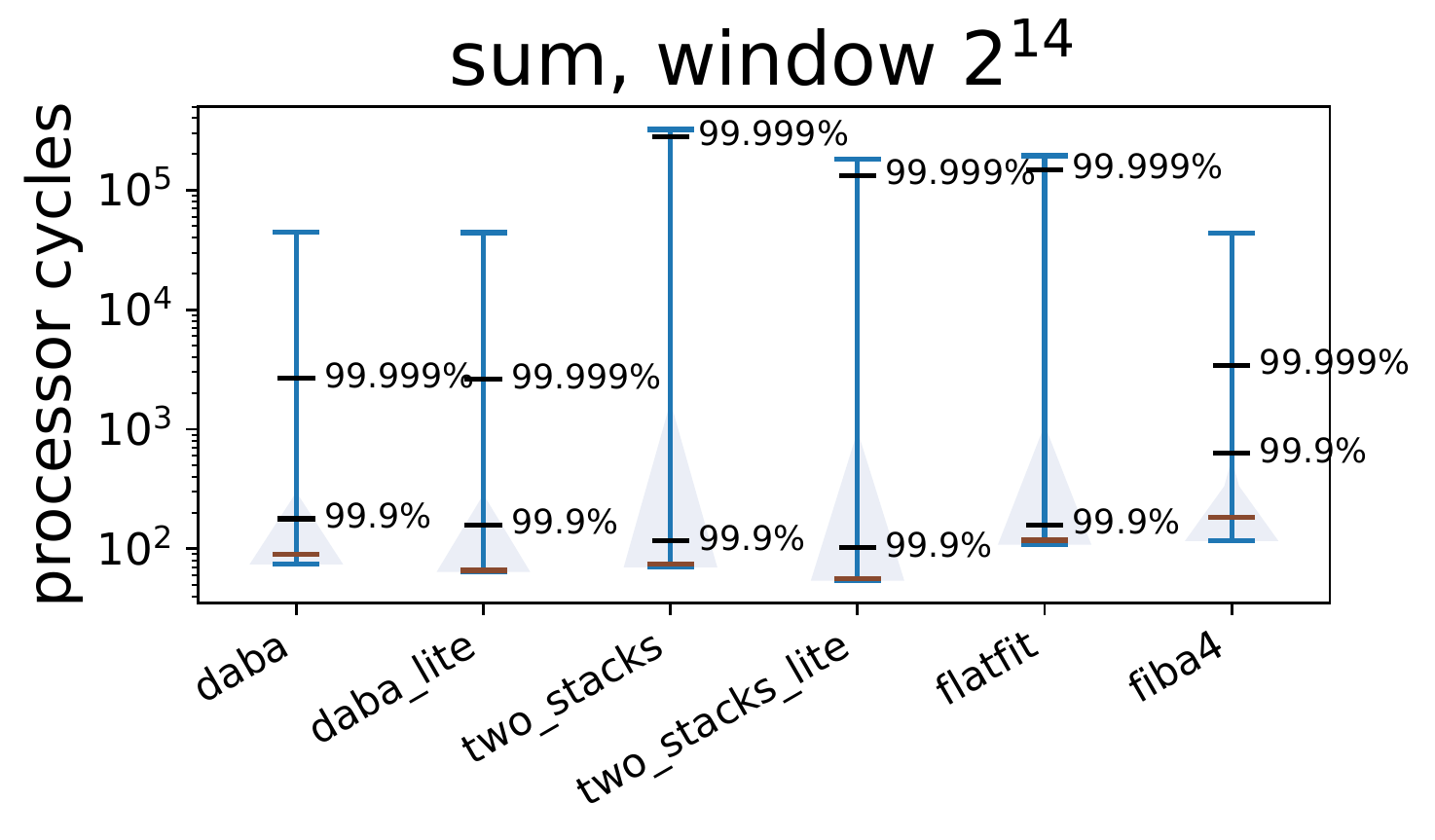}
\includegraphics[width=0.33\textwidth]{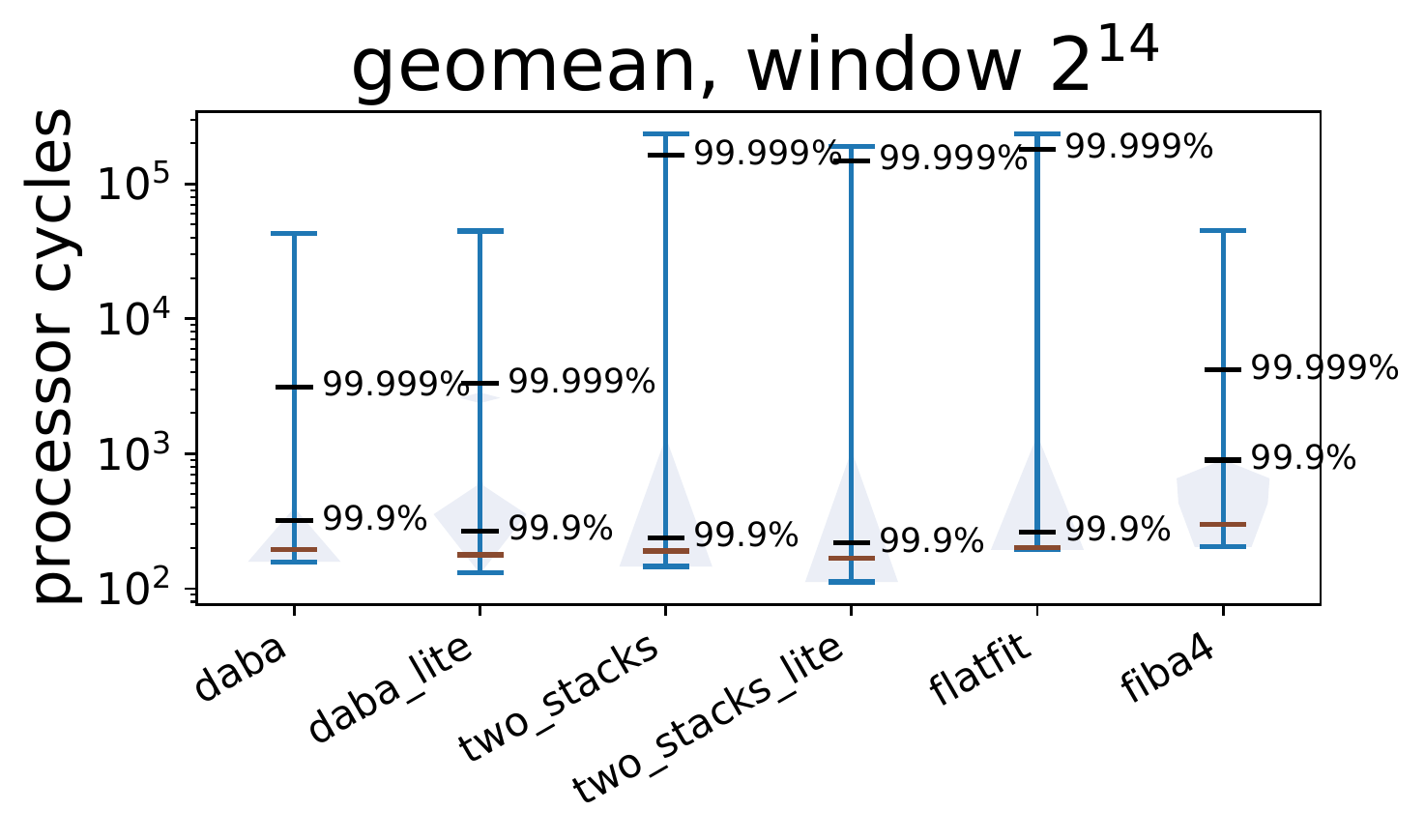}
\includegraphics[width=0.33\textwidth]{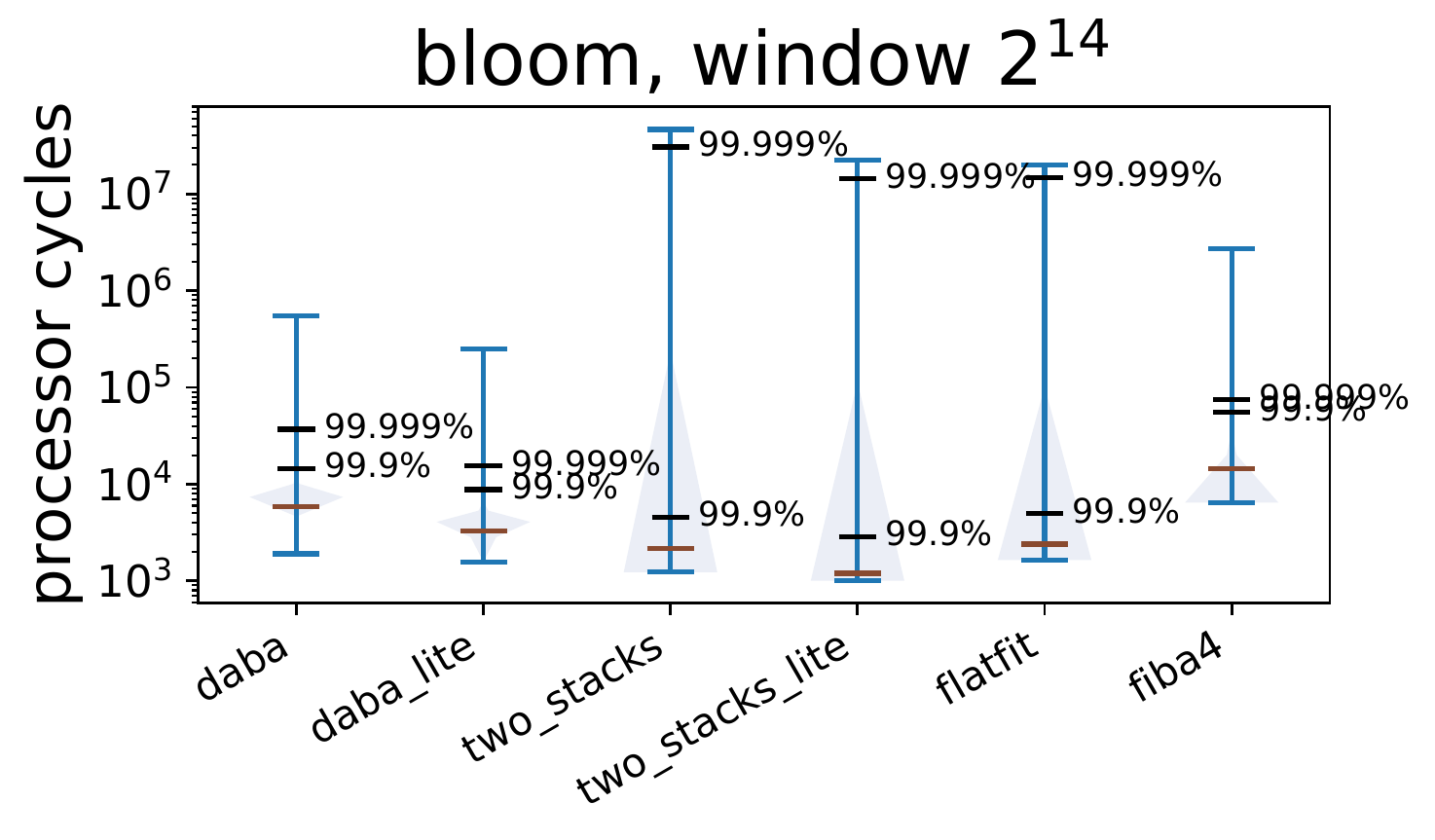}
\caption{Latency, shown as violin plots, for static count-based windows with synthetic data.}
\label{fig_latency}
\end{figure*}

\begin{figure*}[!t]
\center
\includegraphics[width=0.33\textwidth]{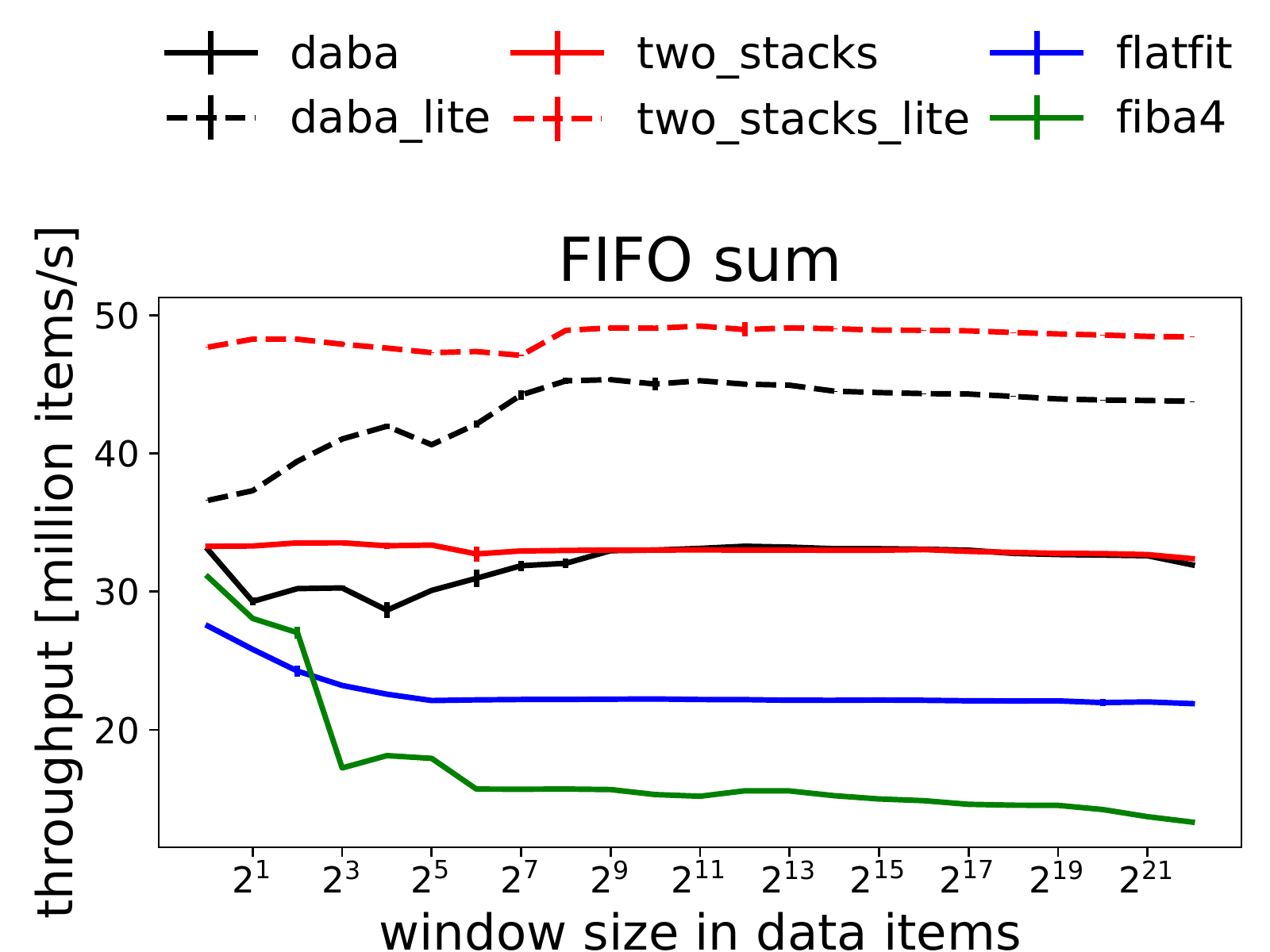}
\includegraphics[width=0.33\textwidth]{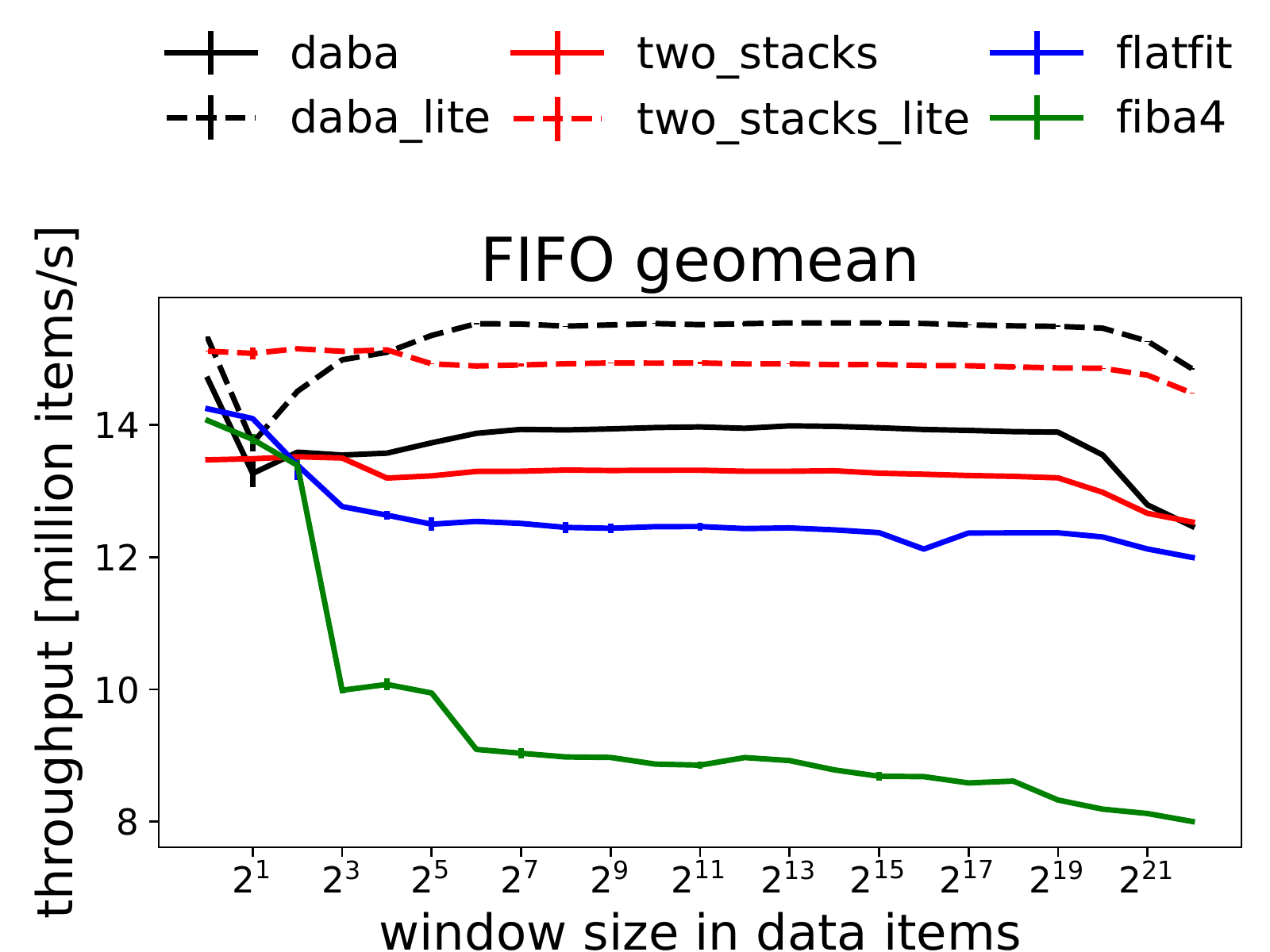}
\includegraphics[width=0.33\textwidth]{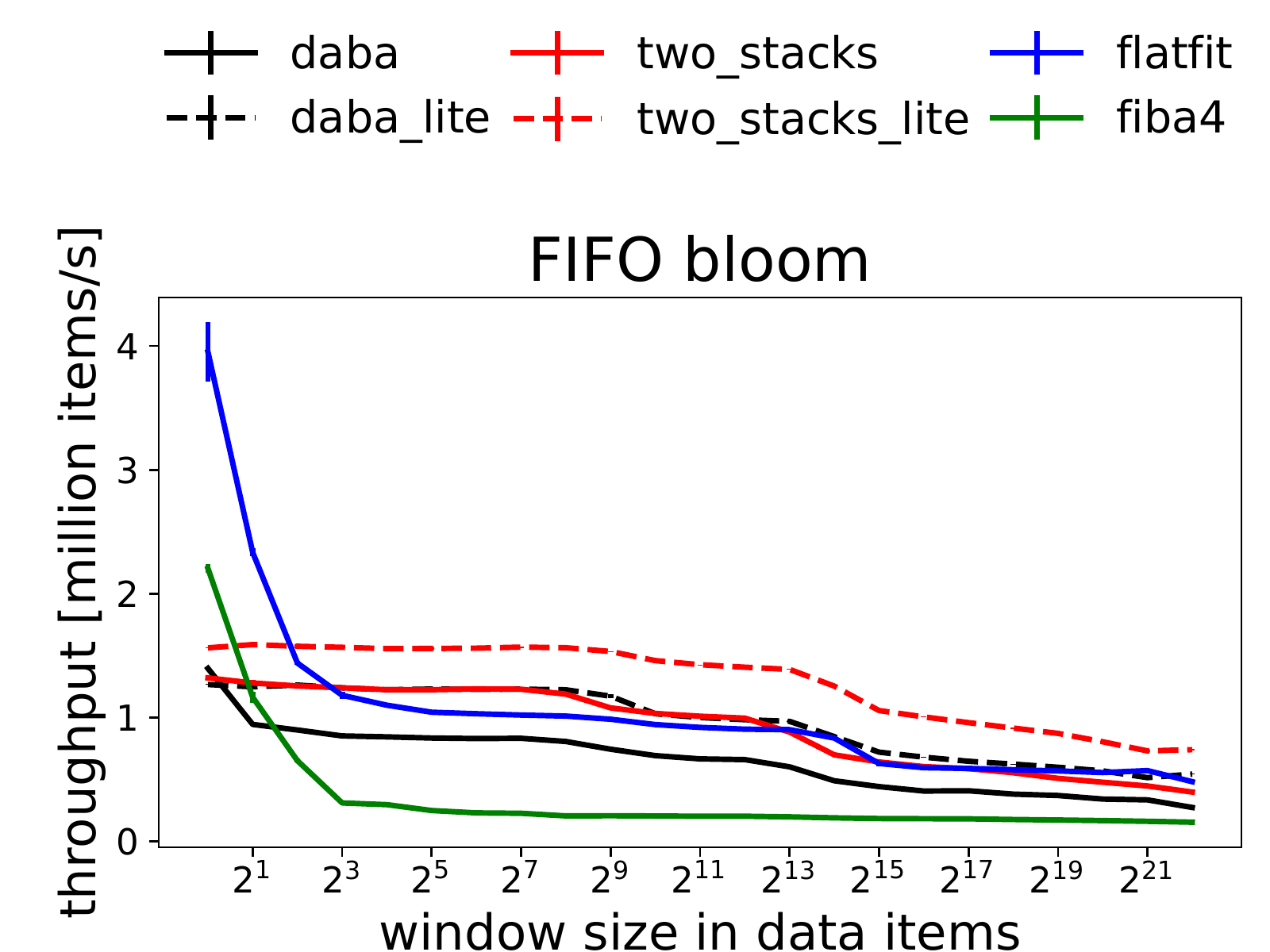}
\caption{Throughput for static count-based windows with synthetic data.}
\label{fig_throughput}
\end{figure*}

The purpose of our experimental evaluation is to test whether DABA's 
worst-case constant algorithmic complexity yields low latency and 
high throughput in practice, and to test whether DABA Lite is always 
more efficient than DABA.

Our experiments use six different SWAGs: Two-Stacks, Two-Stacks Lite, and
FlatFIT~\cite{shein_chrysanthis_labrinidis_2017} are all amortized $O(1)$,
worst-case $O(n)$ algorithms designed for FIFO data. As originally published,
FlatFIT does not support dynamic windows. We use a modified version that
resizes FlatFIT's circular buffer using the standard array doubling/shrinking
technique without disturbing FlatFIT's internal pointer structure.  This
results in an additional amortized $O(1)$ time per operation but supports
dynamic windows and guarantees that the memory footprint is within a constant
factor of the window size.  We also adapted the published FlatFIT algorithm to
our SWAG framework. DABA and DABA Lite are both worst-case $O(1)$ algorithms
designed for FIFO data.  FiBA~\cite{tangwongsan_hirzel_schneider_2019} is
designed for out-of-order data and reduces to amortized $O(1)$, worst-case
$O(\log n)$ in the FIFO case. All of our experiments with FiBA use a min-arity
of 4.

We chose three representative aggregation operators to span the execution
cost spectrum. The operator \typeName{sum} is the sum of all items in the
window, and it represents aggregation operations so cheap that the traversal
and changes to the underlying data structure should dominate performance. The
operator \typeName{bloom} applies a Bloom filter to all items in the window,
and it represents aggregation operations where the operator itself dominates
performance. The operator is so expensive that minimizing calls to it matters
more than changes to the underlying data structure. Finally, \typeName{geomean}
computes the geometric mean of all items in the window. It represents a
middle ground of operator cost.

We implemented all algorithms in C++11, using the g++ compiler version 7.5.0
with optimization level~\lstinline{-O3}.  Our system runs Red Hat 7.3, with
Linux kernel version 3.10.0. The processor is an Intel Platinum 8168 at 2.7
GHz. All implementations, experiments, and post-processing scripts used in 
this section are available from the open-source project Sliding Window Aggregators\footnote{Available at \url{https://github.com/IBM/sliding-window-aggregators}. Our experiments use the C++ 
implementations and benchmarks, as well as the Python scripts from commit \lstinline{41ee775}.}.

\subsection{Static Windows}

The experiments in Figs~\ref{fig_latency} and~\ref{fig_throughput} use a
static count-based window with synthetic data. In each experiment, we first
\procName{insert} $n$ data items, where $n$ is the size of the window. The
timed part of the experiment consists of rounds of \procName{evict},
\procName{insert}, and \procName{query}. For the latency experiments, we record
all times for 10 million rounds with a fixed window size of $2^{14}$ data items.
For the throughput experiments, we time how long it takes to complete 200
million rounds, and we vary $n$ from 1 to $2^{22}$.

The practical reason to choose a worst-case $O(1)$ aggregator is to minimize
latency. Both Two-Stacks and Two-Stacks Lite in Fig~\ref{fig_latency} tend to
have lower minimum latency than both DABA and DABA Lite. But, true to their
linear worst-case, Two-Stacks and Two-Stacks Lite regularly suffer from an
order-of-magnitude higher latency.  This trend becomes
more pronounced as the cost of the aggregation function increases.  Unlike the
other aggregators, FiBA is tree based. As maintaining the tree is more up-front
work, it tends to have high minimum and median latency. But, also being tree
based, its worst-case behavior is bounded by $O(\log n)$; it has lower
worst-case latency than the worst-case $O(n)$ aggregators. FlatFIT is not a
tree-based structure, but the access pattern during queries ends up having
similar properties: successive indirect accesses to different parts of the
window. Each query pushes indices onto a
stack, and then pops indices from the
stack to indirectly access the window.\footnote{Our implementation performs an
optimization where the same stack is reused across queries. This is safe
because the stack is always empty at the end of a query. For dynamic windows,
the number of indices involved can be non-constant. Avoiding the
recreation of the stack and reusing the same memory makes about a 20\%
difference in throughput, but does not change FlatFIT's overall comparative
performance.} This is a large amount of work, and when the aggregation
operation is cheap, this work dominates performance and yields a high latency
floor. 

All of the aggregators are able to maintain close to constant behavior in
Fig~\ref{fig_throughput}, although there are large differences between them.
Surprisingly, DABA's throughput is more competitive with Two-Stacks than in
prior
work~\cite{tangwongsan_hirzel_schneider_2015,tangwongsan_hirzel_schneider_2017}.
We attribute this difference to a more modern compiler with more aggressive
inlining and dead-code elimination. Both DABA Lite and Two-Stacks Lite always
outperform the corresponding non-Lite versions. FlatFIT becomes more competitive with
expensive operations as its indirect accesses become less important compared to
the total number of calls to the aggregation operation.

\begin{figure*}[!t]
\center
\includegraphics[width=0.33\textwidth]{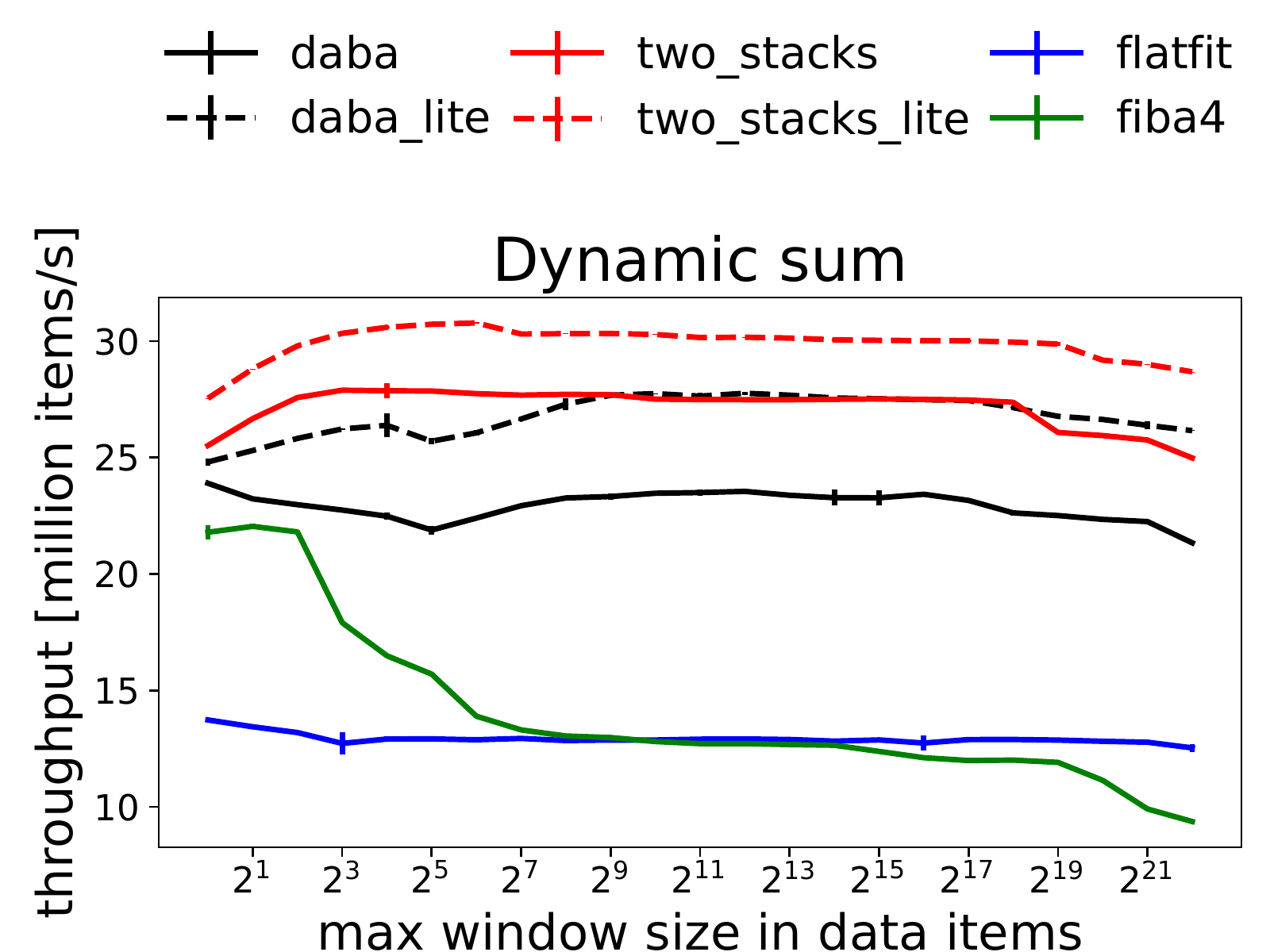}
\includegraphics[width=0.33\textwidth]{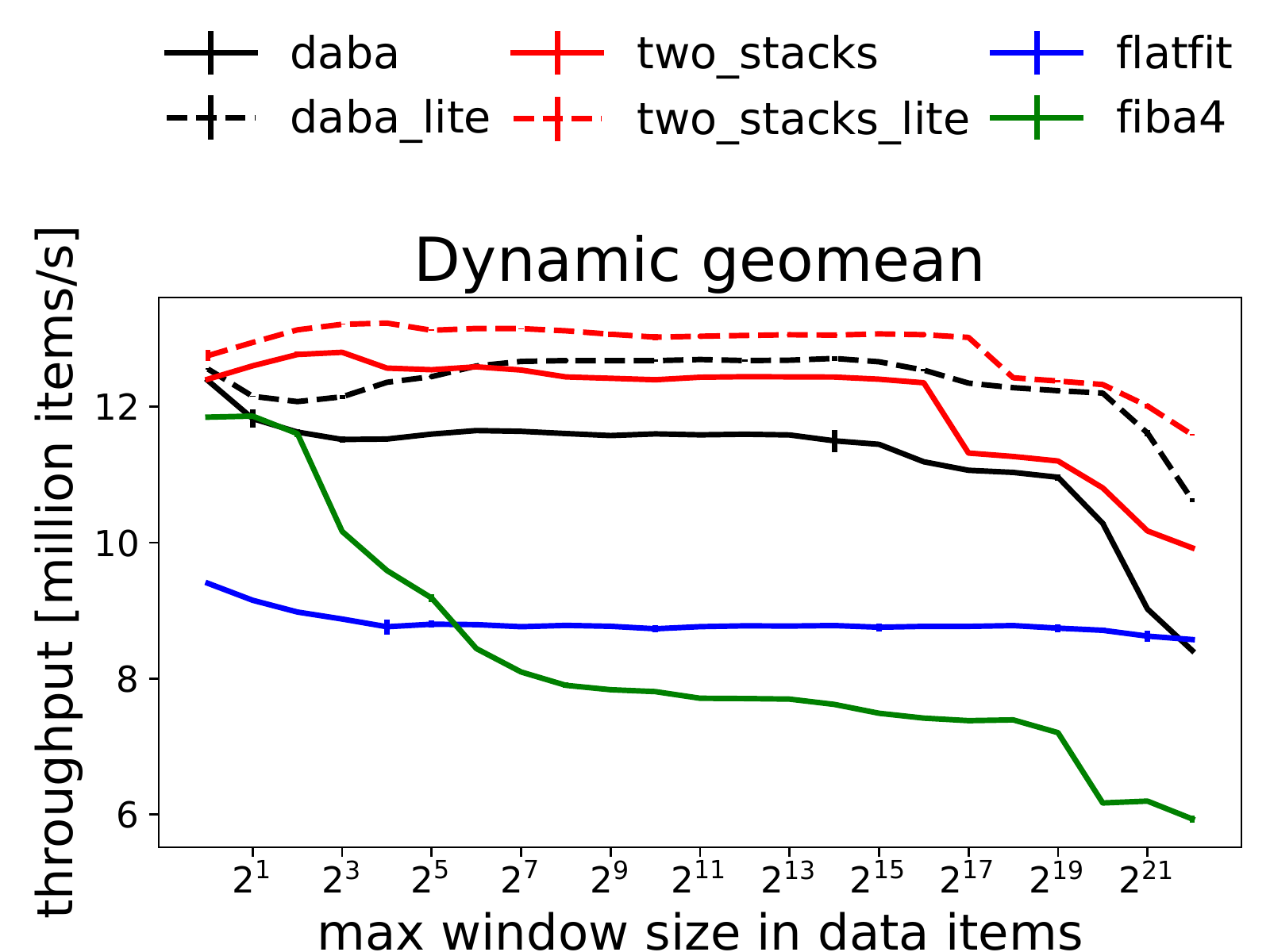}
\includegraphics[width=0.33\textwidth]{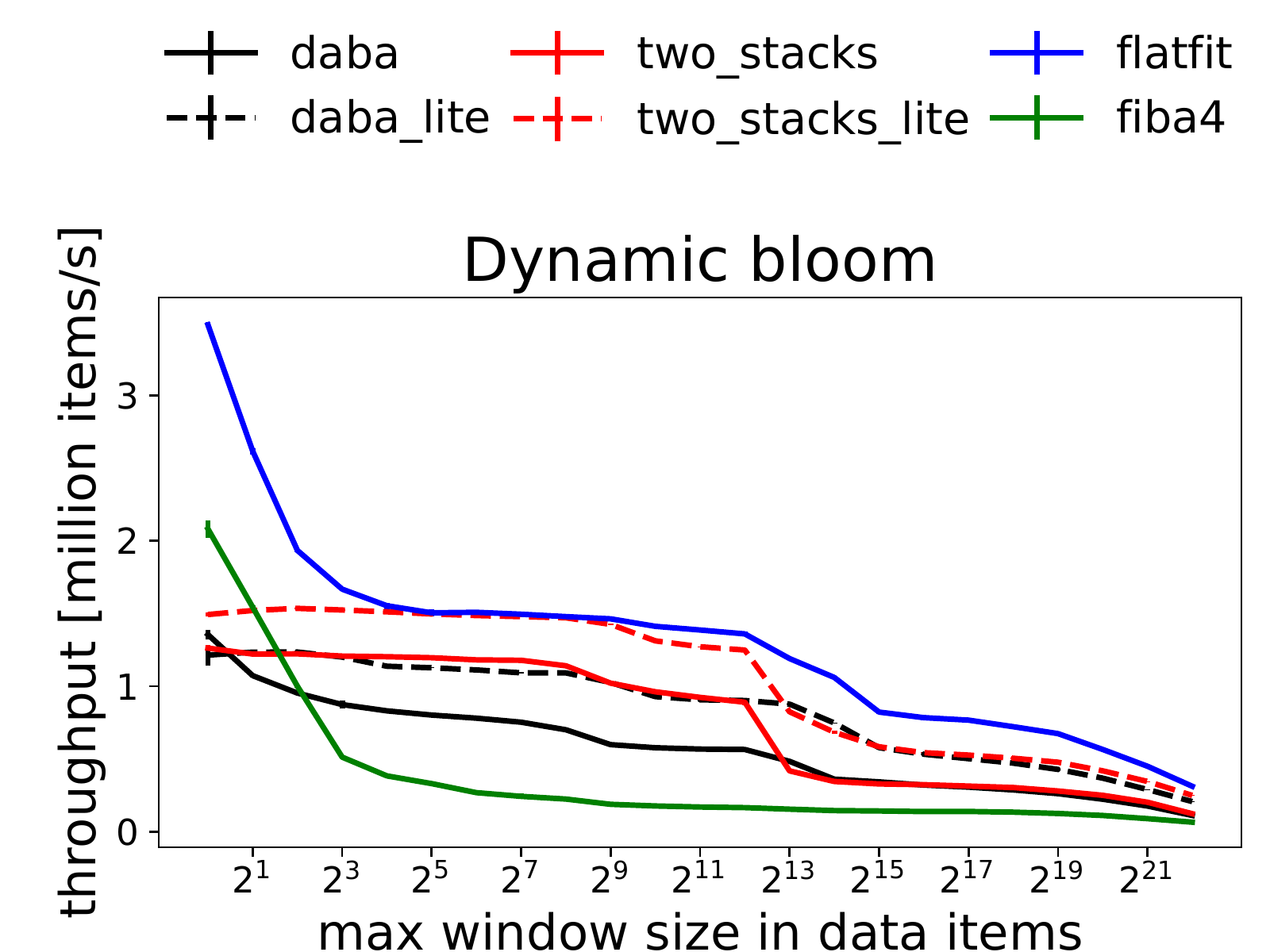}
\caption{Throughput for dynamic count-based windows with synthetic data.}
\label{fig_dynamic}
\end{figure*}

\subsection{Dynamic Windows}

The experiments in Fig.~\ref{fig_dynamic} use a dynamic count-based window
with synthetic data. The experiments time a fill-and-drain pattern for a
total of 200 million data items. It performs \procName{insert} and
\procName{query} until reaching the window size $n$, and then calls
\procName{evict} until the window is down to 0, then repeats. We vary $n$ from 
1 to $2^{22}$. 

The throughput trends are largely the same as with static windows, which is the
point of these experiments: even with dynamically changing window sizes, the
fundamental properties of these streaming aggregation algorithms remain mostly
the same. The one major difference is FlatFIT, whose throughput is consistently
the best for \typeName{bloom}, which is the most expensive aggregation operation. This experimental design
happens to be close to a best-case for FlatFIT, as it does not call the
aggregation operation on evictions. FlatFIT only calls the aggregation
operation on queries. The other algorithms require calling the aggregation
operation on evictions in order to maintain their various properties of their
partial aggregates. But, since the experiment performs no queries when it drains
the window, such work is ``wasted'' in this case.

\subsection{Real Data}

\newcommand*{\twWidth}[0]{\tau}
The experiments in Fig.~\ref{fig_data} use dynamic event-based windows based
on real data. We use the dataset from the DEBS 2012 Grand
Challenge~\cite{debs2012_gc}, which recorded data from manufacturing equipment
at approximately 100 Hz. We removed about 1.5\% of the 32.3 million events to
enforce in-order data to make it suitable for FIFO aggregation algorithms.  Our
experiments maintain an event-time-based window of $\twWidth{}$ seconds, which means that
the actual number of data items in that window will fluctuate over time. We do not
start measuring until the window has evicted its first data item. In the
throughput experiments, we vary $\twWidth$ from 10 milliseconds to 6 hours.  In the
latency experiments, we choose a window of 10 minutes. For both experiments, we
use an aggregation operation inspired by Query 2 of the DEBS 2012 Grand
Challenge: relative variation.

\begin{figure}[t]
\center
\includegraphics[width=0.33\textwidth]{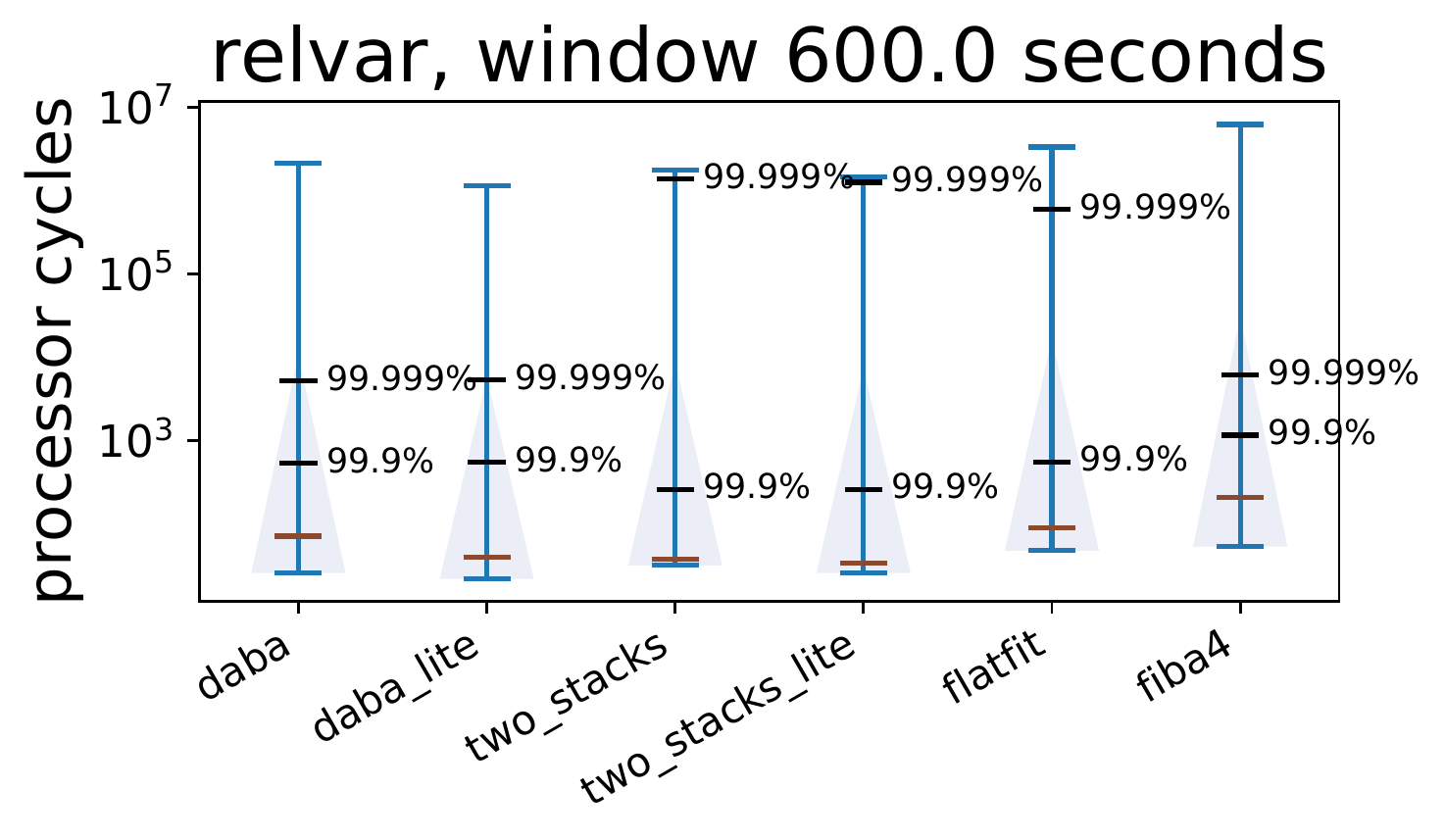} \\
\includegraphics[width=0.33\textwidth]{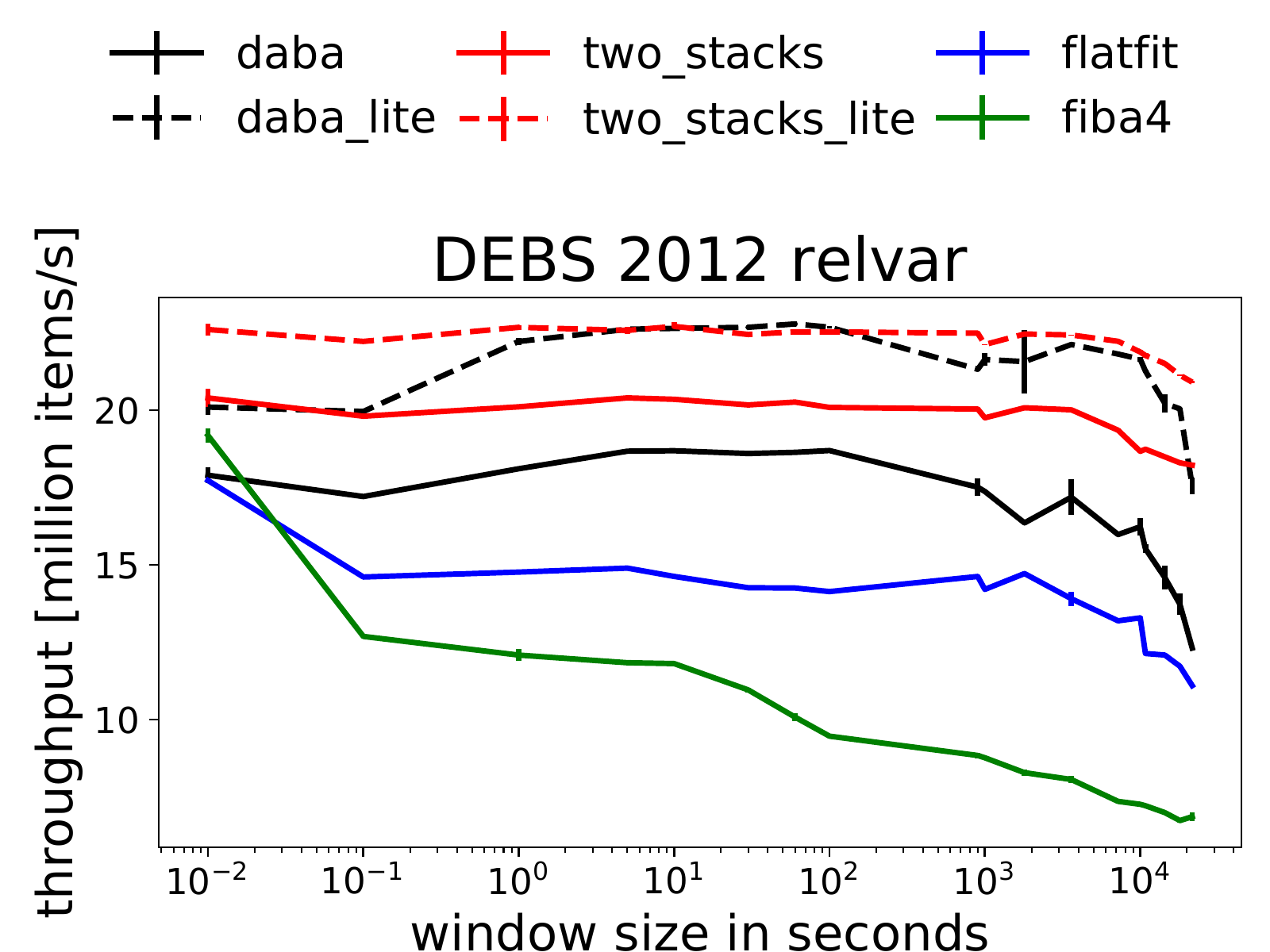}
\caption{Throughput and latency for dynamic event-based windows with the manufacturing 
         equipment data set.}
\label{fig_data}
\end{figure}

Because FiBA was designed for out-of-order data, it natively has a concept of
timestamps. But the other algorithms are FIFO aggregators and do not natively
support timestamps. These sets of experiments use modified versions of all of
the other aggregators that add support to query the youngest and oldest
timestamps in the window. We use those queries to maintain the event-time 
window of $\twWidth{}$ seconds.

The results are consistent with the previous experiments, with two exceptions.
First, DABA, DABA Lite, Two-Stacks, Two-Stacks Lite, and FlatFIT experience
similar maximum latency, unlike with the static experiments. This latency
similarity is caused by rare bulk evictions: rounds where more than 100 items
are evicted experienced latencies more than 100$\times$ greater than the previous
round. Since maintaining the same window is a shared property across all
experiments, all algorithms have similar maximum latency. Second, the
throughput of all algorithms experiences some degradation after $10^3$ seconds.
This is because the actual window size is becoming a large enough fraction of
the total data that rare events are not amortized.

\subsection{Discussion}

Our experiments demonstrate several consistent behaviors. The throughput of the
Lite variants of both DABA and Two-Stacks are always significantly better than
the original version, but the overall difference in latency is less dramatic.
Two-Stacks Lite tends to have the best throughput. Since we lacked
access to an implementation of FlatFIT by its authors, we
re-implemented it based on their paper and extended it to support
variable-sized windows. Our implementation of
FlatFIT tends to have a high latency floor, but a median that is close to that
floor. However, its maximum latency is consistent with being worst-case $O(n)$,
and its throughput is only competitive with expensive operators or when the
experiment happens to align with its design.

The theoretical analysis established that both DABA and DABA Lite are
worst-case constant time algorithms. But low latency and competitive
throughputs are sensitive to what the actual constant is. Our experiments
demonstrate that DABA and DABA Lite are able to realize low latency and
competitive throughput across a large range of $n$ with both static and
dynamic windows and with real data.


\section{Related Work}\label{sec_related}

This section discusses the literature on sliding window aggregation
algorithms with an emphasis on in-order streams and associative
aggregation operators~$\otimes$. As before, let $n$ be the window
size.

\subsection{Solutions to the Same Problem}\label{sec_related_sameproblem}

Section~\ref{sec_problemdef} formalized the problem statement as an
abstract data type called SWAG for first-in-first-out sliding window
aggregation. This section discusses concrete algorithms that implement
the abstract data type.  Section~\ref{sec_problemdef} phrased
associative aggregation operators as monoids. While some monoids are
invertible or commutative, that is not true for all monoids, and thus,
this section only lists algorithms that work without such additional
algebraic properties. This section presents algorithms in
chronological order by publication date.

\emph{Recalculate-from-scratch} implements SWAG by maintaining a FIFO
queue of all data items and calculating the aggregation of the entire
queue for each query. Each query requires $O(n)$ invocations
of~$\otimes$.  The queue takes up space for $n$ data stream values.

The \emph{B-Int} algorithm from 2004~\cite{arasu_widom_2004}
implements SWAG using base intervals, which are similar to a balanced
binary tree. The time complexity is $O(\log n)$ invocations
of~$\otimes$ and the space complexity is around $2n$ partial
aggregates: $n$ for leaves plus \mbox{$\sim n$} for internal nodes.

The \emph{Two-Stacks} idea was mentioned in a Stack Overflow post from
2011~\cite{adamax_2011}, which described the idea for one aggregation
operator, minimum. Even though Two-Stacks is a SWAG algorithm, it was
not immediately noticed as such by the academic community.  As
discussed in Section~\ref{sec:twostacks}, the time complexity is
amortized $O(1)$ invocations of~$\otimes$ with a worst-case of~$O(n)$,
and the space is $2n$ partial aggregates.

The \emph{Reactive Aggregator} from 2015 implements SWAG using a
perfect binary tree. This algorithm uses a data structure called
FlatFAT, which stands for flat fixed-sized aggregator and represents a
perfect binary tree without storing explicit pointers and without
needing any rebalancing. The time complexity is amortized $O(\log n)$
with a worst-case of~$O(n)$ invocations of~$\otimes$.  If $n$ is a
power of two, FlatFAT requires space for $\sim 2n$ partial aggregates:
$n$ leaves plus $n$ interior nodes. If $n$ is slightly above a power
of two, the space can be up to $\sim 4n$.

The \emph{DABA} algorithm was first published in
2015~\cite{tangwongsan_hirzel_schneider_2015}. DABA was inspired by Okasaki's
purely functional queues and deques~\cite{okasaki_1995}. However, the two differ
substantially: DABA is not a purely functional data structure, and Okasaki's
data structures do not implement sliding window aggregation. As discussed in
Section~\ref{sec:daba}, DABA requires worst-case $O(1)$ invocations of~$\otimes$
and the space is $2n$ partial aggregates.

The \emph{FlatFIT} algorithm from 2017 implements SWAG via a flat and
fast index traverser~\cite{shein_chrysanthis_labrinidis_2017}.  The
time complexity is amortized $O(1)$ invocations of~$\otimes$ with a
worst-case of~$O(n)$. The algorithm stores $n$ partial aggregates as
well as $n$ pointers, which are indices into the window for stitching
together the partial aggregates of subranges. The algorithm requires
an additional stack of indices for pointer updates, and the authors
report the total space requirements as up to~$2.5n$.

The \emph{Hammer Slide} paper from 2018~\cite{theodorakis_et_al_2018}
starts from Two-Stacks and optimizes it further. The time complexity
remains amortized $O(1)$ invocations of~$\otimes$ with a worst-case
of~$O(n)$. One of the optimizations from Hammer Slide is to only store
$n+1$ partial aggregates by observing what is needed for the front
stack and the back stack. The Two-Stacks Lite algorithm in
Section~\ref{sec:twostackslite} takes inspiration from Hammer Slide.

The \emph{AMTA} algorithm from
2019~\cite{villalba_berral_carrera_2019} implements SWAG via an
amortized monoid tree aggregator.  AMTA adds sophisticated tree
representations that optimize FIFO insert and evict. Its amortized
algorithmic time complexity is $O(1)$ invocations of~$\otimes$, with a
worst-case of $O(\log n)$. Like other tree-based SWAGs, AMTA requires
$\sim 2n$ space for $n$ leaves and $\sim n$ inner nodes. KVS-AMTA is
an out-of-memory variant that externalizes most of this space into a
key-value store.

The \emph{FiBA} algorithm from
2019~\cite{tangwongsan_hirzel_schneider_2019} implements SWAG via a
finger B-tree aggregator. FiBA uses finger pointers, position-aware
partial aggregates, and a suitable rebalancing strategy to optimize
insert and evict near the start and end of the window. For the FIFO
case, its amortized algorithmic time complexity is $O(1)$ invocations
of~$\otimes$, with a worst-case of $O(\log n)$. Its space complexity
depends on the arity of the B-tree. Since the minimum arity of B-trees
is more than binary, B-trees store fewer than $\sim 2n$ partial
aggregates.

The \emph{DABA Lite} algorithm has not been published before, making
it an original contribution of this paper. As discussed in
Section~\ref{sec:dabalite}, the time complexity is worst-case $O(1)$
invocations of~$\otimes$ and the space is $n+2$ partial aggregates.

\subsection{Complementary Techniques}\label{sec_related_complementary}

While Section~\ref{sec_problemdef} states the core problem for
sliding-window aggregation, there are often additional requirements.
This section discusses techniques for augmenting algorithms that
implement the SWAG abstract data type (including DABA and DABA Lite)
to solve a broader set of problems.

\emph{Coarse-grained sliding} reduces the effective window size $n$ by
storing only a single partial aggregate for values that will be
evicted together. A state-of-the-art algorithm for coarse-grained
sliding is Scotty~\cite{traub_et_al_2018}. Reducing $n$ reduces the
time complexity of any algorithms whose time complexity depends
upon~$n$. Being worst-case $O(1)$, the time complexity of DABA and
DABA Lite does not depend on~$n$.  Reducing $n$ also reduces the space
complexity, which is somewhere between $n$ and $\sim 4n$ for all SWAG
algorithms from Section~\ref{sec_related_sameproblem}.

\emph{Bounded disorder handling} tolerates out-of-order arrivals of
data stream items as long as the disorder is not too large. Srivastava
and Widom described how to handle bounded disorder by buffering
incoming data items~\cite{srivastava_widom_2004}. Later, when data
items are released from the buffer, they are ordered by their nominal
timestamps. That makes it possible to use inorder SWAG algorithms from
Section~\ref{sec_related_sameproblem}.

\emph{Partition parallelism} is a way to parallelize stateful
streaming applications as long as the computation for each partition
key is independent from the computation for the other
keys~\cite{schneider_et_al_2015}. Sliding-window aggregation is often
used in a way that satisfies this requirement, by aggregating
separately within each key. In that case, parallelization can just
maintain separate instances of a given SWAG. For this to work well, it
is best not to conservatively preallocate too much memory, lest the
data structures for rare keys take up too much space.

Given various algorithms and techniques, how can we pick and combine
the right ones for a given problem? A recent paper by Traub et
al.~\cite{traub_et_al_2019} presents decision trees for dispatching to
the right combination given the stream order, window kinds,
aggregation operators, window sharing, etc. We argue that DABA Lite
should be used for the inorder case with associative aggregation
operators.

\subsection{Solutions to Other Problems}\label{sec_related_otherproblems}

Of course there are also problems around sliding-window aggregations
where it does not suffice to just combine a SWAG algorithm from
Section~\ref{sec_related_sameproblem} with a complementary technique
from Section~\ref{sec_related_complementary}. This section highlight a
few such problems with solutions; for more details
see~\cite{hirzel_schneider_tangwongsan_2017}.

\emph{Window sharing} serves sliding window aggregation que\-ries for
multiple window sizes from a single data structure. Not all data
structures are suitable for this. SWAG algorithms from
Section~\ref{sec_related_sameproblem} that support window sharing
include \mbox{B-Int}~\cite{arasu_widom_2004},
FlatFIT~\cite{shein_chrysanthis_labrinidis_2017}, and
FiBA~\cite{tangwongsan_hirzel_schneider_2019}.
The SlideSlide algorithm implements SWAG for fixed-sized
windows~\cite{theodorakis_pietzuch_pirk_2020}.

\emph{Unbounded disorder handling} tolerates out-of-order arrivals
that are arbitrarily late, incorporating them into the data structure
whenever they arrive. Truviso accomplishes this for the case where
multiple input streams have drifted arbitrarily far from each other,
as long as each of the input streams is internally
in-order~\cite{krishnamurthy_et_al_2010}. FiBA supports general
out-of-order sliding window aggregation without restrictions on the
degree of disorder~\cite{tangwongsan_hirzel_schneider_2019}.  The
algorithmic complexity of FiBA matches the theoretical lower bound for
this problem.

When it comes to \emph{aggregation operators}, some algorithms are
more restrictive than our problem statement from
Section~\ref{sec_problemdef}. For instance, subtract-on-evict is a
simple algorithm that only works when subtraction is well-defined, in
other words, when the $\otimes$ operator is invertible.  Similarly,
SlickDeque~\cite{shein_chrysanthis_labrinidis_2018} only works for
aggregation operators that are either invertible or that satisfy the
property that \mbox{$x\otimes{}y\in\{x,y\}$}. On the other hand, there
are also some aggregation operators for which our problem statement
from Section~\ref{sec_problemdef} is a poor fit. One of the most
prominent ones is median, or more generally, percentile
aggregation. An efficient solution for sliding-window median and
percentiles is an order statistics tree~\cite{hirzel_et_al_2016}.



\section{Conclusion}\label{sec_conclusion}

This paper is a journal version of our earlier conference
paper~\cite{tangwongsan_hirzel_schneider_2017} about DABA,
the first algorithm for in-order sliding window aggregation in
worst-case constant time. Besides providing a more comprehensive
description of DABA, this paper also introduces a new
algorithm called DABA Lite that improves over DABA. Where DABA
requires space to store $2n$ partial aggregates, DABA Lite only stores
$n+2$ partial aggregates. Whereas DABA requires on average 2.5
invocations of the underlying monoid per insert and 1.5 per evict,
DABA Lite requires on average only 2 invocations per insert and 1 per
evict. The worst-case time complexity is constant just like for DABA.

DABA and DABA Lite have several desirable properties. They only
require an associative monoid (no need for commutativity nor
invertibility).  They support dynamically-sized windows, where the
window size can fluctuate throughout the execution, for instance, due to
a variable interarrival rate of stream data items.  They are built on
a simple flat data structure, thus avoiding memory-copy or allocation
churn, as well as avoiding excessive pointer chasing. Our experiments
demonstrate that DABA Lite performs well compared to other
sliding-window aggregation algorithms.


\bibliographystyle{spmpsci}      
\bibliography{local}
\balance

\end{document}